\newif\ifecta

\ectafalse  

\ifecta

\documentclass[draft]{ectaart-copy-modified}
\RequirePackage[OT1]{fontenc}
\RequirePackage{graphicx}
\RequirePackage{amsthm}
\RequirePackage[cmex10]{amsmath}
\RequirePackage[round]{natbib}
\RequirePackage[colorlinks]{hyperref}
  \hypersetup{linkcolor=black,filecolor=black,citecolor=black,urlcolor=black}
\RequirePackage{hypernat}
\RequirePackage{amssymb}
\RequirePackage{mathrsfs}
\RequirePackage{enumitem}
\RequirePackage{subfig}
\startlocaldefs
\numberwithin{equation}{section}
\theoremstyle{plain}

\endlocaldefs

\else

\documentclass[11pt,letter]{article}
\RequirePackage[papersize={8.5in,11in},margin=1in]{geometry}
\RequirePackage{graphicx}
\RequirePackage{amsthm}
\RequirePackage{amsmath}
\RequirePackage[round]{natbib}
\RequirePackage[colorlinks]{hyperref}
  \hypersetup{linkcolor=blue,filecolor=blue,citecolor=blue,urlcolor=blue}
\RequirePackage{hypernat}
\RequirePackage{amssymb}
\RequirePackage{mathrsfs}
\RequirePackage{enumitem}
\RequirePackage{subfig}

\fi

\theoremstyle{plain}
\newtheorem{ctheorem}{Theorem}
\newtheorem{theorem}{Theorem}[section]

\newtheorem{lemma}[theorem]{Lemma}

\newtheorem{claim}[theorem]{Claim}

\newtheorem{definition}[theorem]{Definition}
\newtheorem{fact}[theorem]{Fact}

\theoremstyle{definition}

\newtheorem{remark}[theorem]{Remark}

\theoremstyle{remark}

\numberwithin{equation}{section}

\newcommand{\namedref}[2]{\hyperref[#2]{#1~\ref*{#2}}}

\newcommand{\sectionref}[1]{\namedref{Section}{#1}}
\newcommand{\appendixref}[1]{\namedref{Appendix}{#1}}
\newcommand{\theoremref}[1]{\namedref{Theorem}{#1}}
\newcommand{\factref}[1]{\namedref{Fact}{#1}}

\newcommand{\lemmaref}[1]{\namedref{Lemma}{#1}}
\newcommand{\claimref}[1]{\namedref{Claim}{#1}}

\newcommand{\equationref}[1]{\namedref{Eq.}{#1}}

\newcommand{\display}[2]{
\smallskip
\begin{center}
\begin{minipage}{#1\textwidth}
#2
\end{minipage}
\end{center}
\smallskip
}

\newcommand{\DEBUG}[1]{}

\newcommand{\defeq}{\stackrel{\mathrm{def}}{=}}

\newcommand{\parhead}[1]{\smallskip \noindent {\bfseries\boldmath\ignorespaces #1.}\hskip 0.9em plus 0.3em minus 0.3em}

\newcommand{\scparagraph}[1]{\smallskip\noindent\textsc{#1.}}

\def\bqed{\hfill$\blacksquare$}
\def\wqed{\hfill$\square$}

\newcommand{\veps}{\varepsilon}

\newcommand{\nat}{\mathbb{N}}
\newcommand{\real}{\mathbb{R}}

\newcommand{\Class}{\mathscr{C}}

\newcommand{\EXP}{\mathbb{E}}
\newcommand{\TV}{\theta}

\newcommand{\BID}{v}

\newcommand{\SPrice}{\Mech_{\mathsf{2P}}}
\newcommand{\MOPT}{\Mech_{\mathsf{opt}}^{(\Err)}}

\newcommand{\EU}{\EXP u}
\newcommand{\Allo}{A}

\newcommand{\Price}{P}

\newcommand{\pureA}{s}
\renewcommand{\sqcup}{,}
\newcommand{\pureB}{t}
\newcommand{\mixedA}{\sigma}
\newcommand{\mixedB}{\tau}
\newcommand{\Mech}{M}
\newcommand{\MechEA}{\OutFunc^{\Allo}}
\newcommand{\MechEP}{\OutFunc^{\Price}}
\newcommand{\Err}{\delta}

\newcommand{\Int}[2][\Err]{#1[#2]}

\newcommand{\SW}{\mathrm{SW}}

\newcommand{\MSW}{\mathrm{MSW}}

\newcommand{\Players}{[n]}

\newcommand{\Outcomes}{\Omega}

\newcommand{\type}{\theta}
\newcommand{\OutFunc}{F}
\newcommand{\Strats}{S}
\newcommand{\Dist}{\Delta}

\newcommand{\Know}{K}

\newcommand{\Utility}{u}

\newcommand{\Cont}{C}

\newcommand{\Bound}{B}

\newcommand{\CC}[1]{#1}

\newcommand{\zWD}[2]{\overset{\mathrm{w}}{\underset{#1,#2}{\succ}}}
\newcommand{\iWD}[2]{\overset{\mathrm{w}}{\underset{#1,#2}{\prec}}}
\newcommand{\zVWD}[2]{\overset{\mathrm{vw}}{\underset{#1,#2}{\succ}}}

\newcommand{\eUWD}{undominated}

\newcommand{\UWD}{\mathsf{UDed}}

\newcommand{\VWDnt}{\mathsf{Dnt}}

\newcommand{\ffext}{\AlloFunc^{(\delta)}}
\newcommand{\zz}{z}
\newcommand{\DD}{D_{\Err}}
\renewcommand{\mid}{\scriptscriptstyle \Delta}
\renewcommand{\stop}{\scriptscriptstyle \top}
\newcommand{\sbot}{\scriptscriptstyle \bot}
\newcommand{\AlloFunc}{f}
\newcommand{\eWPF}{allocation function}

\newcommand{\defem}[1]{\textbf{\emph{#1}}}

\newcommand{\SPriceF}{\OutFunc_{\mathsf{2P}}}
\newcommand{\SPriceS}{\Strats_{\mathsf{2P}}}

\begin{document}

\ifecta
\else

\title{Knightian Auctions}
\date{\today}
\author{
Alessandro Chiesa \\ \href{mailto:alexch@csail.mit.edu}{alexch@csail.mit.edu} \\ MIT CSAIL
\and
Silvio Micali \\ \href{mailto:silvio@csail.mit.edu}{silvio@csail.mit.edu} \\ MIT CSAIL
\and
Zeyuan Allen Zhu \\ \href{mailto:zeyuan@csail.mit.edu}{zeyuan@csail.mit.edu} \\ MIT CSAIL
}


\maketitle

\begin{abstract}
We study single-good auctions in a setting where each player knows his own valuation only within a constant multiplicative factor $\Err \in (0,1)$, and the mechanism designer knows $\Err$.
The classical notions of implementation in dominant strategies and implementation in undominated strategies are naturally extended to this setting, but their power is vastly different.

On the negative side, we prove that no dominant-strategy mechanism can guarantee social welfare that is significantly better than that achievable by assigning the good to a random player.

On the positive side, we provide tight upper and lower bounds for the fraction of the maximum social welfare achievable in undominated strategies, whether deterministically or probabilistically.

\vfill
\begin{center}
A non-archival draft of the introduction of this paper is presented at \\ the 3rd Innovations in Theoretical Computer Science (ITCS~2012) conference, \\ with the title \emph{Mechanism Design with Approximate Valuations}.
\end{center}
\end{abstract}

\clearpage
\tableofcontents
\clearpage

\fi

\section{Introduction}
\label{sec:introduction}

The goal of this paper is to design mechanisms guaranteeing high social welfare
in auctions of a single good whose players are {\em Knightian}.

\subsection{Knightian Players}

In a traditional single-good auction, each player $i$ is assumed to know his true valuation for the good, $\theta_i$, {\em exactly.}
The assumption, however, may be quite strong. For instance,  can $i$ be really sure that his true valuation is exactly \$17,975 rather than ---say--- \$18,001? If not, then {\em how can his uncertainty be modeled?}

A classical answer is to assume that $i$ knows the single probability distribution $D_i$ from which his true valuation is drawn.
More generally, \cite{Knight21}, and later on \cite{Bewley02}, suggested  to  assume that $i$ knows only a \emph{set} of distributions, one of which is guaranteed to be $D_{i}$.

In an auction, however, from a strategic perspective a ``Knightian player'' could collapse each candidate distribution in his set to its expected value. Accordingly, without loss of generality, in a \emph{Knightian auction} each player $i$ only knows a set of integers, $\Know_{i}$, guaranteed to contain his true valuation $\type_{i}$. Therefore, $\Know_{i}$ is  ``the set of all possible candidates for $\type_{i}$ in $i$'s mind", and will be referred to as $i$'s {\em candidate-valuation set}.

\subsection{Our Knightian Focus}

Knightian players have received much attention in {\em decision theory} or  in mechanisms with a {\em single}  player. %
We are instead interested in studying the competition of {\em multiple} Knightian players in \emph{full-fledged} mechanisms.
Transforming rich (i.e., exact or Bayesian) knowledge into optimal mechanisms is important. But equally important is to understand  whether there are good mechanisms when the players only have set-theoretic knowledge about themselves. 

Specifically, we focus on Knightian auctions of  a single good, adopting
for simplicity sake a finite perspective. Namely,
\begin{itemize}

\item
all valuations will be integers between $0$ and a \emph{valuation bound} $\Bound$, and

\item
all mechanisms specify finitely many pure strategies for each player.

\end{itemize}
(Our results can however be extended to infinite settings as well.)

\subsection{Knightian Mechanism Design}
\label{sec:knowledge-accuracy}

Intuitively, a mechanism cannot perform well in a Knightian setting where the candidate-valuation sets $\Know_i$ are  too ``spread out'', but it might perform well when they are sufficiently ``clustered''. Accordingly, we believe that performance should be measured as a function of the ``inaccuracy" of the players' knowledge.

\parhead{Measuring inaccuracy}
For  a candidate-valuation set $\Know_{i}$ of a player $i$, we set
\begin{equation*}
\hfill
\Err_{i} \defeq \frac{\max \Know_i-\min \Know_i}{\max \Know_i + \min \Know_i} \enspace.
\hfill
\end{equation*}
Then, it is immediately seen that $\Err_i \in [0,1]$ and that, because $\theta_i\in \Know_i$, ``player $i$ knows $\type_{i}$  within a multiplicative factor of $\Err_{i}$".

We  refer to $\Err_i$ as {\em $i$'s  individual inaccuracy} (about his internal knowledge).

We define the {\em (global) inaccuracy} of a Knightian setting to be $\Err \defeq \max_{i} \Err_{i}$.

A setting is \emph{Knightian} if $\Err > 0$ and \emph{traditional} if $\Err=0$.

\noindent
(If $c \in \real$ and $\alpha \in [0,1]$, then  we may call $[c-\alpha c ,c +\alpha c]$    a \emph{$\alpha$-interval} with  \emph{center} $c$. Any set contained in a $\alpha$-interval will be called  \emph{$\alpha$-approximate}.)

\parhead{Designer knowledge}
To study  auction mechanisms in a Knightian setting we must specify what information is available to the designer.

In a Bayesian setting, where the true valuation profile is drawn from a common prior distribution $D$,
it is traditionally assumed that $D$ itself is known to the designer (and the players). In a Knightian setting, each player $i$ knows a set of distributions $\{D_{i,1},\dots,D_{i,k}\}$, which, as already observed, is equivalent to knowing the set  of their expected values, $\Know_{i}$. So: how much of this information should the designer be allowed to know?

An extreme assumption is that  he knows every $\Know_{i}$. A much weaker assumption is that he knows all individual inaccuracies, but not any candidate-valuation set $\Know_i$. A yet weaker assumption  is that he  knows only the maximum individual inaccuracy. This is the assumption we choose to work under: namely, when designing an auction mechanism for  Knightian setting,

\begin{center}
\emph{the designer  only knows the global inaccuracy parameter $\Err$}.
\end{center}

\parhead{Performance}
Being in a tough set-theoretic setting, we adopt a worst-case analysis for evaluating the social welfare performance of an auction mechanism $\Mech$.  Figurately speaking, we envisage the following process. First, an auction mechanism $M$ is announced for selling a given good to $n$ Knightian players, where each player $i$ privately knows his own candidate-valuation set $\Know_i$. Then  ---aware of $\Know \defeq \Know_1 \times \cdots \times \Know_n$ and $\Mech$, and intending to fool $\Mech$--- the devil secretly chooses a true-valuation profile $\type \in \Know$. After that, each player $i$ chooses a strategy $\mixedA_{i}$. (A player $i$ may learn $\type_i$ only after the auction is over, or never, but during the auction he acts based only on $\Know_i$.) Finally, $\Mech$ is played with strategy profile $\mixedA = (\mixedA_{1},\dots,\mixedA_{n})$ so as to produce a winner $w=w(\Mech,\mixedA)$. Note that $w$ is in general a random variable, since every $\sigma_i$ may be mixed and $M$ may be probabilistic.

The \emph{maximum social welfare}  is of course $\max_i \type_i$; and the \emph{realized social welfare} of $\Mech$ on  $\mixedA$ is $\EXP[\type_{w(\Mech,\mixedA)}]$. Informally speaking, the \emph{social welfare performance} of $\Mech$ relative to $\Know$ and $\mixedA$ is taken to be $\min_{ \type \in \Know} \frac{\EXP[\type_{w(\Mech,\mixedA)}]}{\max_i \type_i }$. (Formally, of course, we must specify the solution concept ``behind $\sigma$".)


\parhead{Objectives} In designing  mechanisms in a Knightian setting we  study and  try to maximize their performance as a function of the global inaccuracy $\Err$. In essence, $\Err$ is our chosen Trojan horse for bringing meaningful mechanism design in the Knightian setting. Without paying attention to the quality of the players' knowledge about themselves, one might design elementary mechanisms, but not ``good'' ones.

\parhead{Q\&A}

\begin{itemize}\itemsep -1pt

\item
\emph{Multiplicative or additive accuracy?}
A greater level of generality is achieved by considering two distinct global inaccuracy parameters: a multiplicative one, $\Err^*$, and an additive one, $\Err^+$, leading to the following modified constraint: for all player $i$ there exists $x_{i} \in \real$ such that
\begin{equation*}
\Know_{i} \subseteq \big[(1-\Err^*)x_{i}-\Err^+,(1+\Err^*)x_{i}+\Err^+\big]\cap \{0,\dots,\Bound\} \text{.}
\end{equation*}
All of our theorems hold for this more general condition. For simplicity, however, we consider only one kind of global inaccuracy parameter, and we find the multiplicative one more meaningful.


\item
\emph{Can real $\Err$'s be really large?} Absolutely.
The players' candidate-valuation sets  may indeed be ``very approximate". Consider a firm participating to an auction for an exclusive license to manufacture solar panels in the US for a period of 25 years. Even if the demand were precisely known in advance, and the only uncertainty came from the firm's ability to lower its costs of production via some breakthrough research, a firm's individual inaccuracy about its own true valuation for the license could easily exceed $0.5$.

\end{itemize}

\subsection{Solution Concepts}

The analysis of every mechanism requires an underlying solution concept.  As Knightian settings are settings of \emph{incomplete information} (i.e., settings whose players do not know exactly the true valuations of their opponents), two solution concepts naturally apply: implementation in  dominant strategies and implementation in undominated strategies.
Of course, both solution concepts need to be properly extended to our setting, but this is naturally done (and in fact done in a way consistent with
all prior works).

In essence,  a pure strategy $s_i$ of a player $i$ is (very weakly) {\em Knightian-dominant} if it provides $i$ with a utility at least as large as that of any other strategy $t_i$ of $i$, no matter what strategies his opponents may choose, and no matter what candidate in $\Know_i$ may be $i$'s true valuation. A pure strategy $s_i$ of $i$ is {\em Knightian-undominated} if $i$ does not have any other strategy $t_i$  that

(1)  gives $i$ utility at least as great as  $s_i$ no matter what strategy subprofiles  his opponents may use, and no matter what member of $\Know_i$ may be $i$'s true valuation, and

(2)  gives $i$ utility strictly greater than $s_i$ for at least some strategy subprofile of his opponents and some member of $\Know_i$.

(The set of such undominated strategies under a mechanism $\Mech$ is denoted by $\UWD_{i}^{\Mech}(\Know_i)$, or simply by $\UWD_{i}(\Know_i)$ when $\Mech$ is clear from context.)

\subsection{Informal Discussion of Our Results}
\label{sec:results}

How much social welfare can we guarantee in auctions? In traditional ones the answer is trivial: $100\%$ in (very-weakly) dominant strategies, via the second-price mechanism. Things are quite different in Knightian auctions.

\subsubsection{Dominant-Strategy Mechanisms Are Meaningful but Inadequate}

Although a Knightian player ``does not have a best valuation to bid", very-weakly-dominant strategies continue to be meaningfully defined in a Knightian auction.
In a traditional auction the revelation principle (see \cite{Mye81}) guarantees that, as far as very-weakly-dominant strategies are concerned, it suffices to consider mechanisms that restrict a player's strategies to (reporting) single valuations. It is easy to see, however, that a  natural exentension of the revelation principle continues to apply in Knightian auctions. Specifically, if a very-weakly-dominant strategy mechanism $\Mech$ with a given social welfare performance guarantee exists, then there also exists a \emph{Knightian-direct} mechanism $\Mech'$, with the same performance, where, for every player $i$, (1) his pure strategy set consists of reporting \emph{sets} of valuations, and (2) truthfully reporting his own candidate-valuation set $\Know_{i}$ is very-weakly dominant.

In principle, therefore, there may be a dominant-strategy mechanism that  obtains all true candidate-valuation sets, $\Know_1,\ldots,\Know_n$, and guarantees a high social welfare performance. Of course, given the inaccuracy of the players' knowledge of their own true valuations,  one should expect some degradation of performance relative to the exact-valuation setting. However, one might conjecture that, in a Knightian auction with global inaccuracy $\Err$, a dominant-strategy mechanism might be able to guarantee some $\Err$-dependent fraction ---such as $(1-\Err)$, $(1-3\Err)$, or $(1-\Err)^2$--- of the maximum social welfare. We prove, however, that even such  modest hopes are overly optimistic.

\begin{ctheorem}[informal]
\label{thm:intro-VWDnt-negative}
For all $n \geq 1$, $\Err \in (0,1)$, and $\Bound > \frac{3-\Err}{2\Err}$, no (possibly probabilistic) very-weakly-dominant-strategy-truthful mechanism $\Mech_{n,\Err,\Bound}$ can guarantee a fraction of the maximum social welfare greater than
\begin{equation*}
\hfill
\frac{1}{n} + \frac{\lfloor \frac{3-\Err}{2\Err} \rfloor +1}{\Bound}
\hfill
\end{equation*}
in any Knightian auction with $n$ players, valuation bound $\Bound$, and inaccuracy parameter $\Err$.
\end{ctheorem}


As a \emph{relative} measure of the quality of the players' self knowledge, $\Err$ should be
independent of the magnitude of the players' valuations. But to ensure an upper bound on the players' valuations, $\Bound$ should be  large. Accordingly, the above result essentially implies that any very-weakly-dominant-strategy mechanism can only guarantee a fraction $\approx \frac{1}{n}$ of the maximum social welfare. However, such a fraction can be trivially achieved by the ``naive'' very-weakly-dominant-strategy mechanism that, dispensing with all bids, assigns the good to a random player! Thus, \theoremref{thm:intro-VWDnt-negative} essentially says that no dominant-strategy mechanism can be smart: ``the optimal one can only be as good as good as the stupid naive one". In other words,

\display{0.9}{
\emph{dominant strategies are \defem{intrinsically} linked to each player having \defem{exact} knowledge of
either (1) his own valuation, or (2) the \defem{unique} possible distribution from which his own valuation has been drawn.}
}

\noindent
By showing the limitations of dominant strategies in Knightian auctions, \theoremref{thm:intro-VWDnt-negative} opens the door to alternative solution concepts: in particular, to implementation in undominated strategies. We actually believe that the Knightian setting will provide a new and vital role for this natural and non-Bayesian implementation notion.


\subsubsection{The Power of Deterministic Undominated-Strategy Mechanisms}
\label{sec:results:undominated}

We tightly characterize the power of implementation in undominated strategies via deterministic mechanisms in Knightian auctions. First of all, without much difficulty, we show that the second-price mechanism (although no longer dominant-strategy) guarantees a relatively good fraction of the maximum social welfare in undominated strategies, despite the fact that it does not leverage any information about $\Err$. Second,  more importantly and perhaps more surprisingly, we prove that no deterministic undominated-strategy mechanism can do better, even with full knowledge of $\Err$.

\parhead{The (good) performance of the second-price mechanism}

\begin{ctheorem}[informal]
\label{thm:intro-VCG-in-UWD-approx-1}
In any Knightian auction with $n$ players, valuation bound $\Bound$, and inaccuracy parameter $\Err$, the second-price mechanism guarantees a fraction of the maximum social welfare that is
\end{ctheorem}
\begin{center}
$\approx \left(\frac{1-\Err}{1+\Err}\right)^{2}
\enspace.$%
\footnote{We note that when breaking ties at random, the performance  of the second-price mechanism is only marginally better: namely, it guarantees a fraction of the maximum social welfare {\em exactly equal} to
 $(\frac{1-\Err}{1+\Err})^{2}$.
}
\end{center}

The course intuition behind \theoremref{thm:intro-VCG-in-UWD-approx-1} is clear:

\medskip

\begin{minipage}[h]{0.9\columnwidth}
\vspace{2mm}
\emph{``It is obvious that each player $i$ should only consider bidding a value $v_{i}$ inside his own candidate-valuation set $ \Know_{i}$. It is further obvious that the worst possible gap between the maximum and the actual social welfare is achieved in the following case. Let $w$ be the winner in the second-price mechanism, and let $h$, $h\neq w$, be the player with the largest \textbf{candidate} valuation. Player $w$ bids $v_w = \max \Know_w$, and player $h$ bids $v_h = \min \Know_h$ (and $v_w$ only slightly exceeds $v_h$). In this case it is obvious that the second-price mechanism guarantees at most a fraction $\approx \big(\frac{1-\Err}{1+\Err}\big)^2$ of the maximum social welfare.\qed"}
\vspace{2mm}
\end{minipage}

\medskip

\noindent
Of course, things are a bit more complex. In particular, the fact that a player $i$ should only consider bids in $\Know_{i}$ (actually more precisely between $\min \Know_{i} - 1$ and $\max \Know_{i} +1 $) requires a proof.

\parhead{The optimality of the second-price mechanism}

\begin{ctheorem}[informal]
\label{thm:intro-VCG-in-UWD-approx-2}
For all $n \geq 2$, $\Err \in (0,1)$, and $\Bound \geq \frac{5}{\Err}$, no deterministic undominated-strategy mechanism $\Mech_{n,\Bound, \Err}$ can guarantee a fraction of the maximum social welfare greater than
\begin{equation*}
\hfill\left(\frac{1-\Err}{1+\Err}\right)^{2} + \frac{4}{\Bound} \hfill
\end{equation*}
in any Knightian auction with $n$ players, valuation bound $\Bound$, and inaccuracy parameter $\Err$.
\end{ctheorem}

\theoremref{thm:intro-VCG-in-UWD-approx-2} is harder to prove, as is to be expected from an impossibility result.
Indeed,  its statement applies to all undominated-strategy mechanisms, so that the revelation principle is no longer relevant. Thus,  to prove \theoremref{thm:intro-VCG-in-UWD-approx-2}, rather than analyzing a single mechanism (the ``direct truthful" one), in principle we should consider \emph{all} possible mechanisms. Considering only those where a player's strategies  consist of  valuations, or even sets of valuations, is not sufficient.  We would have to consider mechanisms with arbitrary strategy sets. Establishing \theoremref{thm:intro-VCG-in-UWD-approx-2} thus requires new techniques, informally discussed in \sectionref{sec:tools}, and formally provided in \sectionref{sec:proof-undominated-intersection-lemma}.
\subsubsection{The Greater Power of Probabilistic Undominated-Strategy Mechanisms}

The second-price mechanism ``ignores the global inaccuracy parameter". It simply guarantees a fraction $\approx (\frac{1-\Err}{1+\Err})^{2}$ of the maximum social welfare in any Knightian auction, no matter what the value of  $\Err$ happens to be. It is thus legitimate to ask whether knowing $\Err$ (or a close upper-bound to it)
enables one to  design
mechanisms with better efficiency guarantees. We prove that this is indeed the case: we explicitly construct a probabilistic mechanism that, by properly leveraging $\Err$, outperforms the second-price mechanism, and then we prove that our mechanism is essentially optimal.

\begin{ctheorem}[informal]
\label{thm:intro_uwd_single_pos_ran-1}
For all $n \geq 2$, $\Err \in (0,1)$, and $\Bound$, there exists a mechanism $\MOPT$ that guarantees a fraction of the maximum social welfare that is at least
\begin{equation*}
\hfill
\frac{(1-\Err)^2 + \frac{4\Err}{n}}{(1+\Err)^2}
\hfill
\end{equation*}
in any Knightian auction with $n$ players, valuation bound $\Bound$, and inaccuracy parameter $\Err$.
\end{ctheorem}


\parhead{Theoretical significance}
\theoremref{thm:intro_uwd_single_pos_ran-1} highlights a  novelty of the Knightian setting: namely, probabilism enhances the power of implementation in undominated strategies even for guaranteeing social welfare. By contrast, probabilism offers no such advantage in the exact-valuation world, since the deterministic second-price mechanism already guarantees maximum social welfare. We conjecture that, in Knightian settings, probabilistic mechanisms will enjoy a provably better performance in other applications as well.

\parhead{Practicality}
The proof of \theoremref{thm:intro_uwd_single_pos_ran-1} is  the technically hardest one in this paper.
Nonetheless, we would like to emphasize that  $\MOPT$ is very \emph{practically} played, as it  requires almost no computation from the players, and a very small amount of computation from the mechanism.
In addition, its performance is \emph{practically} preferable to that of the second-price mechanism. For instance, when $\Err=0.5$, $\MOPT$ guarantees a social welfare that is at least \textbf{five} times higher than that of the second-price mechanism when there are $2$ players, and at least \textbf{three} times higher when there are $4$ players.
(For a full comparison chart, see \appendixref{sec:diagrams}.)

\bigskip
If $\MOPT$ proves the power of probabilistic mechanisms, our next  theorem upperbounds this power by proving that the social-welfare performance of $\MOPT$ is essentially optimal among all mechanisms, probabilistic or not.

\begin{ctheorem}[informal]
\label{thm:intro_uwd_single_pos_ran-2}
For all $n \geq 2$, $\Err \in (0,1)$, and $\Bound \geq \frac{5}{\Err}$, no (possibly probabilistic) undominated-strategy mechanism $\Mech_{n,\Err,\Bound}$ can guarantee a fraction of the maximum social welfare greater than
\begin{equation*}
\hfill \frac{(1-\Err)^2 + \frac{4\Err}{n}}{(1+\Err)^2} + \frac{4}{\Bound} \hfill
\end{equation*}
in any Knightian auction with $n$ players, valuation bound $\Bound$, and inaccuracy parameter $\Err$.
\end{ctheorem}


In sum, our results prove that mechanism design in the Knightian setting is quite possible. Some of the old techniques no longer work, but it is still possible to  construct good mechanisms.

\subsection{Two Techniques of Independent Interest}
\label{sec:tools}

New ventures require new tools. Let us thus  highlight two techniques,  crucial to
our present endeavor,  that we believe  will prove useful also to future work in Knightian mechanism design.

\parhead{The Undominated Intersection Lemma}
To prove  \theoremref{thm:intro-VCG-in-UWD-approx-2} and \theoremref{thm:intro_uwd_single_pos_ran-2}, we  establish a basic structural relation between candidate-valuation sets and undominated strategies. The simplest one of course would be $\UWD_{i}(\Know_{i})=\Know_{i}$.
This relation, however, is generally false, even when the strategies available to each player consist of
individual valuations between $0$ and $\Bound$.%
\footnote{Indeed a mechanism does not need to interpret a bid $v_{i}$ reported by $i$ as $i$'s true valuation $\TV_{i}$. For instance, the mechanism could first replace each $v_{i}$ by $\pi(v_{i})$ where $\pi$ is some fixed permutation over $\{0,1,\dots,\Bound\}$ and then run the second-price mechanism as if each player $i$ had bid $\pi(v_{i})$. In this case, after $\UWD(\Know_{i})$ has been correctly computed, it will look very different from $\Know_{i}$.}
A second relation, implied by the previous one, is the following:
\begin{center}

$\forall$ {\em   $\Know_i$  and  $\widetilde{\Know_i}$,} $\Know_{i}\cap \widetilde{\Know}_{i}\neq \emptyset \Rightarrow \UWD_{i}(\Know_{i}) \cap \UWD_{i}(\widetilde{\Know}_{i})\neq \emptyset$.

\end{center}
It is not clear, however, whether this second relation always holds.%
\footnote{It would actually hold if the total number of coins usable by the players for choosing their mixed strategies were upper-bounded by a fixed constant.}
Indeed, an undominated-strategy mechanism may have to specify its strategy sets in quite unforseen ways.
Therefore, as soon as $\Know_{i}$ and $\widetilde{\Know}_{i}$ are even slightly different, their corresponding $\UWD_{i}(\Know_{i})$ and $\UWD_{i}(\widetilde{\Know}_{i})$ may in principle be totally unrelated. We prove, however, that the following simple variation of the second relation holds for any possible mechanism. Informally,
\begin{itemize}
\item[]
\emph{For any mechanism, probabilistic or not, if $\Know_{i}$ and $\widetilde{\Know}_{i}$ have at least \textrm{{\bf two}} values in common, then there exist two (possibly mixed) ``almost payoff-equivalent'' strategies $\mixedA_{i}$ and $\widetilde{\mixedA}_{i}$ respectively having $\UWD_{i}(\Know_{i})$ and $\UWD_{i}(\widetilde{\Know}_{i})$ as their support.}
\end{itemize}
This relation actually  suffices for deriving all our impossibility results.

\parhead{The Distinguishable Monotonicity Lemma}
To prove that a given social choice function can be implemented in undominated strategies we are happy to consider mechanisms using a restricted kind of strategies and allocation functions, but we must achieve a delicate balance. On one hand, these restrictions should ensure that the undominated strategies corresponding to a given candidate-valuation set can be characterized in a way that is both \emph{conceptually simple} and \emph{easy to work with}. On the other hand, they should be {\em sufficient} for proving our
\theoremref{thm:intro-VCG-in-UWD-approx-1} and \theoremref{thm:intro_uwd_single_pos_ran-1}.

Specifically, we consider mechanisms whose strategies consist of individual valuations (i.e., the pure strategies of each player coincide with $\{0,\ldots,\Bound\}$) and whose allocation functions are restrictions (to $\{0,\dots,\Bound\}^n$) of integrable functions (over $[0,\Bound]^n$) satisfying a suitable monotonicity property. A simple lemma, the \emph{Distinguishable Monotonicity Lemma}, then guarantees that, for all candidate-valuation set $\Know_i$,
\begin{center}
$\UWD_{i}(\Know_{i})=\{\min \Know_{i}, \max\Know_i\}$.
\end{center}
Although concerned with undominated strategies, when applied to the case of players knowing their valuations exactly, the Distinguishable Monotonicity Lemma is a strengthening of a classical lemma  characterizing (very weakly) dominant-strategy-truthful mechanisms in traditional single-good auctions.
Further, the Distinguishable Monotonicity Lemma actually applies to \emph{all} single-parameter domains, not just single-good auctions (the same way that the classical lemma does).
We thus believe that this simple lemma will be useful beyond the immediate needs of this paper.

\section{Prior Work with Knightian Players}

As already mentioned, Knightian players have received a lot of attention in decision theory. In particular, \cite{Incomplete1},  \cite{Incomplete2}, \cite{Incomplete-vector1} and \cite{Incomplete-vector2} investigate decision with incomplete orders of preferences.

The merits of different ways for a Knightian player to ``condense'' his set of possible values into a meaningfully and deterministically-chosen \emph{single value} have been explored. For example, \cite{Incomplete-random} studies the average, \cite{Incomplete-expect2} the maximum, and \cite{Incomplete-expect1}  the Choquet expectation.

Other authors have studied mechanisms where a \emph{single} Knightian player is called to accept or reject a given offer; in particular \cite{LopomoRS09} studied the rent-extraction problem in such a setting.

Less relevant to our work, several authors have considered individual Bayesians to model a player's uncertainty: for instance, \cite{Sandholm00}, \cite{PRST08}, and \cite{FT11}. Also, others have studied equilibrium models with unordered preferences: for instance \cite{knight-equilibrium1}, \cite{knight-equilibrium2}, \cite{knight-equilibrium3}, and \cite{knight-equilibrium4}. More recently, \cite{RigottiShannon05} have characterized the set of equilibria in a financial market problem.

\section{Single-Good Knightian Auctions}
\label{sec:auction}

We separate every  auction into two parts, a context and a mechanism.

\parhead{Knightian Contexts}
A Knightian context  $\Cont$  has the following components.%
\footnote{The players' beliefs are not part of our contexts because they cannot affect our strong solution concept.}
\begin{itemize}[nolistsep]

  \item $\Players = \{1,2,\dots,n\}$, the set of \emph{players.}

  \item $\{0,1,\dots,\Bound\}$, the set of all  {\em  valuations}, where $\Bound$ is the {\em valuation bound}.

  \item $\Err \in (0,1]$, the {\em  inaccuracy } of the context.

  \item $\Know$, the profile of {\em candidate-valuation sets}, where, for all $i$, $\Know_i \subseteq \Err [x_{i}]$ for some  $x_i \in \real$.
Here,  $\Int{x} \defeq [(1-\Err)x ,(1+\Err)x] \cap \{0,1,\dots,\Bound\}$.

  \item $\TV$, the profile of \emph{true valuations}, where each $\TV_{i} \in \Know_i$.
  \item $\Outcomes=\{0,1,\ldots,n\} \times \real^n$, the set of \emph{outcomes}. If $(a,\Price)\in\Outcomes$, then we refer to $a$ as an \emph{allocation} and to $P$ as a profile of \emph{prices}. (If $a=0$ then the good remains unallocated, else player $a$ wins the good.)
  \item $\Utility$, the profile of \emph{utility functions.} Each $\Utility_{i}$ maps  any  outcome $(a,\Price)$ to  $\TV_{i}-\Price_{i}$ if $a=i$, and to $-\Price_{i}$ otherwise.

\end{itemize}
\smallskip

\noindent
Notice that $\Cont$ is fully specified by $n,\Bound,\Err,\Know$, and $\TV$, that is, $\Cont=(n,\Bound,\Err,\Know,\TV)$.

\scparagraph{Knowledge}
In a context $C=(n,\Bound,\Err,\Know,\TV)$ each player $i$ only knows $\Know_{i}$ and that $\TV_{i} \in \Know_{i}$,  and a mechanism designer only knows $n$, $\Bound$, and $\Err$.

\scparagraph{Notation}
 The set of all  contexts with $n$ players, valuation bound $\Bound$, and inaccuracy $\Err$, is denoted by $\Class_{n,\Bound,\Err}$.

The \emph{social welfare} of an outcome $(a,P)$ relative to true-valuation profile $\TV$, $\SW(\TV,(a,\Price))$, is defined to be $\TV_{a}$.
The \emph{maximum} social welfare of a true-valuation profile $\TV$, $\MSW(\TV)$, is defined to be $\max_{i \in \Players} \TV_{i}$.

\parhead{Mechanisms} Our mechanisms  for Knightian contexts are  finite and ordinary.
Indeed a mechanism  for $\Class_{n,\Bound,\Err}$ is a pair $\Mech = (\Strats,\OutFunc)$ where
\begin{itemize}[nolistsep]
  \item $\CC{\Strats} = \Strats_{1}\times\dots\times\Strats_{n}$ is the set of all {\em pure strategy profiles} of $\Mech$, and
  \item $\OutFunc \colon \Strats \to \{0,1,\ldots,n\} \times \real^{n}$ is $\Mech$'s \emph{outcome function}.
\end{itemize}
Set $S$ is always finite and non-empty, and function $F$ may be probabilistic.

\scparagraph{Notation}
\begin{itemize}[nolistsep]
\item We denote pure strategies by Latin letters, and possibly mixed strategies by Greek ones.
\item If $\Mech=(\Strats,\OutFunc)$ is a mechanism and $s\in \Strats$, then by  $\MechEA_{i}(\pureA)$ and $\MechEP_{i}(\pureA)$ we respectively denote the probability that the good is assigned to player $i$ and the expected price paid by $i$ under strategy profile $\pureA$. For mixed strategy profile $\mixedA\in\Dist(\Strats)$, we define $\MechEA_{i}(\mixedA)\defeq \mathbb{E}_{\pureA\gets \mixedA}\big[\MechEA_{i}(\pureA)\big]$ and $\MechEP_{i}(\mixedA)\defeq \mathbb{E}_{\pureA\gets\mixedA}\big[\MechEP_{i}(\pureA)\big]$, where $\pureA \gets \mixedA$ denotes that ``$\pureA$ is drawn from the mixed strategy $\mixedA$''.
\item We refer to  $\MechEA$ as the \emph{\eWPF{}} \emph{of $\Mech$}. More generally,  we say that  $\AlloFunc \colon \Strats \to [0,1]^{n}$ is an \emph{\eWPF{}} if  for all strategy profile $\pureA\in\Strats$, $\sum_{i \in \Players} \AlloFunc_{i}(\pureA) \leq 1$.
\end{itemize}


\section{Knightian Dominance}
\label{sec:dominance}

In extending  the three classical notions of dominance to our approximate valuation setting, the obvious constraint is that when each candidate-valuation  set $\Know_{i}$ consists of a single element, then all extended notions must collapse to the original ones.

\begin{definition}
\label{def:dominance}
In a mechanism $\Mech=(\Strats,\OutFunc)$ for  $\Class_{n,B,\Err}$, let $K$ be a profile of candidate-valuation  sets,  $i$ a player, $\mixedA_i$ a (possibly mixed) strategy of $i$, and $s_i$ a pure strategy of $i$. Then, relatively to $\Know_i$,  we say that
\begin{itemize}[nolistsep]
  \item $\mixedA_{i}$ \defem{very-weakly dominates} $\pureA_{i}$, in symbols $\mixedA_{i} \zVWD{i}{\Know_{i}} \pureA_{i}$, if

$
\forall\, \TV_{i} \in \Know_{i} \,,\,
\forall\, \pureB_{-i} \in \CC{\Strats}_{-i} \,:\,
\EU_{i}(\TV_{i} , \OutFunc(\mixedA_{i} \sqcup \pureB_{-i}) ) \geq \EU_{i}(\TV_{i} , \OutFunc(\pureA_{i} \sqcup \pureB_{-i}) ) \enspace.
$

\item $\mixedA_{i}$ \defem{weakly dominates} $\pureA_{i}$, in symbols $\mixedA_{i} \zWD{i}{\Know_{i}} \pureA_{i}$, if $\quad
(a) \; \; \mixedA_{i} \zVWD{i}{\Know_{i}} \pureA_{i} \text{ and }$

$(b) \; \;
 \exists\, \TV_{i} \in \Know_{i} \,,\,
\exists\, \pureB_{-i} \in \CC{\Strats}_{-i} \,:\,
\EU_{i}(\TV_{i} , \OutFunc(\mixedA_{i} \sqcup \pureB_{-i}) ) > \EU_{i}(\TV_{i} , \OutFunc(\pureA_{i} \sqcup \pureB_{-i}) ) \enspace.
$

\end{itemize}
The  \defem{(very-weakly-)dominant strategies} for $\Know_i$ and $K$ respectively are
  \begin{equation*}
  \VWDnt_{i}(\Know_{i}) \defeq \left\{ \pureA_{i} \in \Strats_{i} : \forall\, \pureB_{i} \in \Strats_{i} \,,\, \pureA_{i} \zVWD{i}{\Know_{i}} \pureB_{i} \right\} \text{and}\; \CC{\VWDnt}(\CC{\Know})\defeq \VWDnt_{1}(\Know_{1}) \times \cdots \times \VWDnt_{n}(\Know_{n})\enspace .
  \end{equation*}
The \defem{undominated strategies} for $\Know_i$ and $K$ respectively are
\begin{equation*}
\UWD_{i}(\Know_{i}) \defeq \left\{ \pureA_{i} \in \Strats_{i} : \not\exists \, \mixedA_{i} \in \Delta(\Strats_{i})\,,\,   \mixedA_{i} \zWD{i}{\Know_{i}} \pureA_{i} \right\} \text{and}\;
\CC{\UWD}(\CC{\Know})\defeq \UWD_{1}(\Know_{1}) \times \cdots \times \UWD_{n}(\Know_{n}).
\end{equation*}

\end{definition}

We use the notation $\VWDnt_{i}$ instead of, say, ``$\mathsf{VWDnt}_{i}$'' because we have no need to define weakly-dominant or  strictly-dominant strategies in a Knightian setting.
(Indeed, Theorem 1 shows that even   \emph{very}-weakly-dominant strategies  cannot guarantee
any non-trivial performance.)%

We use the notation $\UWD_{i}$ instead of, say, ``$\mathsf{UWDed}_{i}$'', because (as in the classical setting) implementation in undominated strategies is defined for weak dominance.%
\footnote{Also, when considering instead \emph{very-weak} dominance, two ``equivalent'' strategies may eliminate each other and the set $\UWD$ may become empty.}


The above extensions of the classical notions are quite straightforward. Only to the extension of weak dominance might require some attention.%
\footnote{
Consider defining ``$\mixedA_{i}$ weakly dominates $\pureA_{i}$'' using the following alternative quantifications in the additional condition for weak dominance: (1) $\forall \TV_{i} \forall \pureB_{-i}$, (2) $\exists \TV_{i} \forall \pureB_{-i}$, and (3) $\forall \TV_{i} \exists \pureB_{-i}$. Alternatives 1 and 2 do not yield  the classical notion of weak dominance when $\Know_{i}$ is singleton. Alternative 3 fails to capture the ``weakest condition'' for which, in absence of special beliefs,  a strategy $\pureA_{i}$ should be discarded in favor of $\mixedA_{i}$.
Indeed, since we already know that $\mixedA_{i}$ very-weakly dominates $\pureA_{i}$, for player $i$ to discard strategy $\pureA_{i}$ in favor of $\mixedA_{i}$, it should suffice that $\pureA_{i}$ is strictly worse than $\mixedA_{i}$ for a \emph{single} true-valuation candidate $\TV_{i} \in \Know_{i}$. That is, we should not insist that $\pureA_{i}$ be strictly worse than $\mixedA_{i}$ for all $\TV_{i} \in \Know_{i}$.}


Finally, let us note that the following obviously holds.

\begin{fact}
\label{fact:UWD-non-empty}
$\UWD_{i}(\Know_{i}) \neq \emptyset$ for all $\Know_{i}$.
\end{fact}

\section{Formal Statement and Proof of Theorem 1}
\label{sec:dominant-strategies}

\textbf{Theorem 1.}\hspace{2mm}
\emph{ For all $ n\geq 1$, $  \Err \in (0,1) $, $  \Bound > \frac{3-\Err}{2\Err}$, and all
 (possibly probabilistic) very-weakly-dominant-strategy-truthful mechanisms
$\Mech=(\Strats,\OutFunc)$ for $\Class_{n,\Bound,\Err}$,  there exists a
context $(n,\Bound,\Err,\Know,\type)\in \Class_{n,\Bound,\Err}$ such that
\begin{equation*}
\EXP\Big[ \SW\big(\TV,\OutFunc({\Know})\big) \Big]
\leq
\left(\frac{1}{n} + \frac{\lfloor \frac{3-\Err}{2\Err} \rfloor +1}{\Bound}\right)\MSW(\TV) \enspace.
\end{equation*}}

\noindent\emph{Proof.}
Fix arbitrarily $n$, $\Err$, and $\Bound$ such that  $\Bound>\frac{3-\Err}{2\Err}$.
We start by proving a separate claim. Essentially, as soon as
a player reports a $\Err$-interval whose center is sufficiently high,  his winning probability  remains constant. (Actually the same holds for his price, although we do not care about it.)

\medskip

\begin{claim}
\label{claim:VDnt-single-item-lemma}
 For all players $i$, all integers $x \in (\frac{3 - \Err}{2\Err},\Bound]$, and all subprofiles $\widetilde{\Know}_{-i}$ of $\Err$-approximate candidate-valuation sets,
 \begin{center}
 $\MechEA_{i}(\Int{x} \sqcup \widetilde{\Know}_{-i})  = \MechEA_{i}(\Int{x+1}\sqcup \widetilde{\Know}_{-i})$.
 \end{center}
\end{claim}


\begin{proof}[Proof of \claimref{claim:VDnt-single-item-lemma}]
Because  $\Know_{i}$  may coincide with $\Int{x}$, and because when this is the case  reporting $\Int{x}$ very-weakly dominates reporting $\Int{x+1}$, the following inequality must hold:  $\forall\, \TV_{i}\in\Int{x}$,
\begin{equation}\label{firstinequality}
\MechEA_{i}(\Int{x} \sqcup \widetilde{\Know}_{-i})\cdot \TV_{i} - \MechEP_{i}(\Int{x}\sqcup \widetilde{\Know}_{-i})
\geq\;
\MechEA_{i}(\Int{x+1} \sqcup \widetilde{\Know}_{-i}) \cdot \TV_{i} - \MechEP_{i}(\Int{x+1} \sqcup \widetilde{\Know}_{-i})
\enspace
\end{equation}

Because $\Know_{i}$  may coincide with $\Int{x+1}$, and because when this is the case  reporting  $\Int{x+1}$ very-weakly dominates reporting $\Int{x}$, the following inequality  also holds: $\forall\, \TV_{i}'\in\Int{x+1}$,
\begin{equation}\label{secondinequality}
\MechEA_{i}(\Int{x+1} \sqcup \widetilde{\Know}_{-i})\cdot \TV_{i}' - \MechEP_{i}(\Int{x+1}\sqcup \widetilde{\Know}_{-i})
\geq\;
\MechEA_{i}(\Int{x} \sqcup \widetilde{\Know}_{-i}) \cdot \TV_{i}' - \MechEP_{i}(\Int{x} \sqcup \widetilde{\Know}_{-i})
\enspace.
\end{equation}
Thus, setting $\TV_{i}=x$ in  \equationref{firstinequality} and $\TV_{i}'=x+1$ in  \equationref{secondinequality}, and summing up
(the corresponding terms of) the resulting inequalities, the $\MechEP_{i}$ price terms and a few other terms cancel out yielding the following inequality:
\DEBUG{
\begin{align*}
&\MechEA_{i}(\Int{x} \sqcup \widetilde{\Know}_{-i})\cdot x + \MechEA_{i}(\Int{x+1} \sqcup \widetilde{\Know}_{-i})\cdot (x+1) \\
\geq\;&
\MechEA_{i}(\Int{x+1} \sqcup \widetilde{\Know}_{-i}) \cdot x + \MechEA_{i}(\Int{x} \sqcup \widetilde{\Know}_{-i}) \cdot (x+1)
\enspace,
\end{align*}
which implies that
}
\begin{equation}\label{thirdinequality}
\MechEA_{i}(\Int{x+1} \sqcup \widetilde{\Know}_{-i}) \geq \MechEA_{i}(\Int{x} \sqcup \widetilde{\Know}_{-i}) \enspace.
\end{equation}
Also, setting $\TV_{i}=\lfloor x(1+\Err) \rfloor$ in \equationref{firstinequality} and $\TV_{i}'=\lceil (x+1)(1-\Err) \rceil$ in \equationref{secondinequality},%
\footnote{The hypothesis $x>\frac{3-\Err}{2\Err}$ implies that $x>\frac{1}{2\Err}$, which in turn implies that,
 under the above choices,  $\TV_{i}\in\Int{x}$ and $\TV_{i}'\in\Int{x+1}$.}
and summing up the resulting inequalities we obtain the following one:
\DEBUG{
\begin{align*}
&\MechEA_{i}(\Int{x} \sqcup \widetilde{\Know}_{-i})\cdot \lfloor x(1+\Err) \rfloor + \MechEA_{i}(\Int{x+1} \sqcup \widetilde{\Know}_{-i})\cdot \lceil (x+1)(1-\Err) \rceil \\
\geq\;&
\MechEA_{i}(\Int{x+1} \sqcup \widetilde{\Know}_{-i}) \cdot \lfloor x(1+\Err) \rfloor + \MechEA_{i}(\Int{x} \sqcup \widetilde{\Know}_{-i}) \cdot \lceil (x+1)(1-\Err) \rceil \enspace,
\end{align*}
which implies that
}
\begin{equation}\label{intermediateinequality}
\Big(\MechEA_{i}(\Int{x} \sqcup \widetilde{\Know}_{-i}) -\MechEA_{i}(\Int{x+1} \sqcup \widetilde{\Know}_{-i}) \Big)
\cdot
\Big(\lfloor x(1+\Err) \rfloor - \lceil (x+1)(1-\Err) \rceil\Big)
\geq 0
\enspace.
\end{equation}
Now notice that $\lfloor x(1+\Err) \rfloor - \lceil (x+1)(1-\Err) \rceil>0$, because, by hypothesis, $x>\frac{3-\Err}{2\Err}$. Thus from \equationref{intermediateinequality} we deduce
\begin{equation}\label{fourthinequality}
\MechEA_{i}(\Int{x} \sqcup \widetilde{\Know}_{-i}) \geq \MechEA_{i}(\Int{x+1} \sqcup \widetilde{\Know}_{-i}) \enspace
\end{equation}
Together, \equationref{thirdinequality} and \equationref{fourthinequality} imply our claim.%
\end{proof}

\medskip

Let us now finish the proof of \theoremref{thm:intro-VWDnt-negative}.
Choose the profile of candidate-valuation  sets $\widehat{\Know} \defeq (\Int{c},\Int{c},\dots,\allowbreak\Int{c})$, where $c \defeq \lfloor \frac{3-\Err}{2\Err} \rfloor + 1$. By averaging, because the summation of $\MechEA_{i}(\widehat{\Know})$ over $i \in \Players$ cannot be greater than $1$, there must exist a player $j$ such that $\MechEA_{j}(\widehat{\Know})\leq 1/n$. Without loss of generality, let such player be player 1. Then, invoking
\claimref{claim:VDnt-single-item-lemma}
multiple times
we have
\begin{equation*}
\MechEA_{1}(\Int{\Bound},\Int{c},\dots,\Int{c}) = \MechEA_{1}(\Int{\Bound-1},\Int{c},\dots,\Int{c})=\cdots = \MechEA_{1}(\Int{c},\Int{c},\dots,\Int{c}) = \MechEA_{1}(\widehat{\Know}) \leq \dfrac{1}{n} \enspace.
\end{equation*}
Now suppose that the true candidate-valuation profile of the players is ${\Know} \defeq (\Int{\Bound},\Int{c},\dots,\allowbreak\Int{c})$. Then, $\TV=(\Bound,c,\dots,c) \in  {\Know}$ and
\begin{equation*}
     \EXP \big[\SW\big(\TV,\OutFunc( {\Know})\big)\big]
\leq \frac{1}{n}\Bound + \frac{n-1}{n}c
\leq \left(\frac{1}{n} + \frac{c}{\Bound}\right)\Bound
= \left(\frac{1}{n} + \frac{c}{\Bound}\right) \cdot \MSW(\TV)
\enspace,
\end{equation*}
as desired.\bqed

\section{The Undominated Intersection Lemma}
\label{sec:proof-undominated-intersection-lemma}

\begin{lemma}[\textbf{Undominated Intersection Lemma}]
\label{lemma:uwd_single_lemma_cycle}
Let $\Mech=(\Strats,\OutFunc)$ be a mechanism, $i$ a player, and $\Know_{i}$ and $\widetilde{\Know}_{i}$ two candidate-valuation  sets of  $i$
such that $|\Know_i \cap \widetilde{\Know_i}|>1$. Then, for every $\veps>0$, there exist mixed strategies $\mixedA_{i} \in \Dist(\UWD_{i}(\Know_{i}))$ and $\widetilde{\mixedA}_{i} \in \Dist(\UWD_{i}(\widetilde{\Know}_{i}))$ such that $\forall\, \pureA_{-i} \in \Strats_{-i}$,
\begin{center}
$\left\vert \MechEA_{i}(\mixedA_{i} \sqcup \pureA_{-i}) - \MechEA_{i}(\widetilde{\mixedA}_{i} \sqcup \pureA_{-i}) \right\vert < \veps$ .
\end{center}
\end{lemma}

(Actually the same holds for $\MechEP$, although we do not care about it.)


\noindent
\emph{Proof.}
Let $x_{i}$ and $y_{i}$ be two distinct integers in $\Know_{i} \cap \widetilde{\Know}_{i}$, and, without loss of generality, let $x_{i}>y_{i}$.

Recall that, by \factref{fact:UWD-non-empty}, $\UWD_{i}(\Know_{i})$ and $\UWD_{i}(\widetilde{\Know}_{i})$ are both nonempty.
If there exists a common (pure) strategy $\pureA_{i} \in \UWD_{i}(\Know_{i}) \cap \UWD_{i}(\widetilde{\Know}_{i})$, then setting $\mixedA_{i}=\widetilde{\mixedA}_{i}=\pureA_{i}$ completes the proof. Therefore, let us assume that  $\UWD_{i}(\Know_{i})$ and $\UWD_{i}(\widetilde{\Know}_{i})$ are disjoint, and let $\pureA_{i}$ be a strategy in $\UWD_{i}(\Know_{i})$ but not in $\UWD_{i}(\widetilde{\Know}_{i})$. The finiteness of the strategy set $\Strats_{i}$ implies the existence of a strategy $\widetilde{\mixedA}_{i} \in \Dist(\UWD_{i}(\widetilde{\Know}_{i}))$ such that $\widetilde{\mixedA}_{i} \zWD{i}{\widetilde{\Know}_{i}} \pureA_{i}$.\footnote{When $\Strats_{i}$ is not finite, we need  of course to assume that the mechanism is bounded; see \cite{Jac92}.}
We now prove that
\begin{center}
\label{eq:int-lemma-prove}
$\exists \, \mixedB_{i} \in \Dist(\UWD_{i}(\Know_{i})) \text{ such that } \mixedB_{i} \zWD{i}{\Know_{i}} \widetilde{\mixedA}_{i}$.%
\footnote{Note that, while we have only defined what it means for a \emph{pure} strategy to be dominated by a possibly mixed one, the definition trivially extends to the case of dominated strategies that are \emph{mixed}, as for  ``$\mixedB_{i} \zWD{i}{\Know_{i}} \widetilde{\mixedA}_{i}$'' in the case at hand.}
\end{center}
Let $\widetilde{\mixedA}_{i} = \sum_{j\in X} \alpha^{(j)} \widetilde{s}_{i}^{(j)}$, where $X$ is a  subset of $\Strats_{i}$. Invoking again the disjointness of the two \eUWD{} strategy sets, we deduce that for each $j\in X$ there exists a strategy $\mixedB_{i}^{(j)} \in\Dist(\UWD_{i}(\Know_{i}))$ such that $\mixedB_{i}^{(j)} \zWD{i}{\Know_{i}} \widetilde{\pureA}_{i}^{(j)}$. Then it is easily seen that $\mixedB_{i} \defeq \sum_{j \in X} \alpha^{(j)} \mixedB_{i}^{(j)}$ is a mixed strategy as desired.

For the same reason, we can also find some $\widetilde{\mixedB}_{i}\in \Dist(\UWD_{i}(\widetilde{\Know}_{i}))$ such that $\widetilde{\mixedB}_{i} \zWD{i}{\widetilde{\Know}_{i}} \mixedB_{i}$. Continuing in this fashion, ``jumping'' back and forth between $\Dist(\UWD_{i}(\Know_{i}))$ and $\Dist(\UWD_{i}(\widetilde{\Know}_{i}))$, we obtain an infinite chain of not necessarily distinct strategies, $\{ \mixedA_{i}^{(k)}, \widetilde{\mixedA}_{i}^{(k)}\}_{k \in \nat}$ (where ($\mixedA_{i}^{(1)}=s_i$, $\widetilde{\mixedA}_{i}^{(1)}=\sigma_i$, and) such that
\begin{center}
$\mixedA_{i}^{(1)}
\iWD{i}{\Know_{i}}
\widetilde{\mixedA}_{i}^{(1)}
\iWD{i}{\widetilde{\Know}_{i}}
\mixedA_{i}^{(2)}
\iWD{i}{\Know_{i}}
\widetilde{\mixedA}_{i}^{(2)}
\iWD{i}{\widetilde{\Know}_{i}}
\cdots$
\end{center}
Since weak dominance implies very-weak dominance, we have that for all $\pureA_{-i}\in\Strats_{-i}$ and all $k \in \nat$:
\begin{equation*}
\begin{array}{lclc}
\forall \widetilde{\TV}_{i} \in \widetilde{\Know}_{i}, & \MechEA_{i}(\mixedA_{i}^{(k)} \sqcup \pureA_{-i}) \widetilde{\TV}_{i} - \MechEP_{i}(\mixedA_{i}^{(k)} \sqcup \pureA_{-i}) &
   \leq & \MechEA_{i}(\widetilde{\mixedA}_{i}^{(k)} \sqcup \pureA_{-i}) \widetilde{\TV}_{i} - \MechEP_{i}(\widetilde{\mixedA}_{i}^{(k)} \sqcup \pureA_{-i}) \label{eqn:uwd_single_lemma_cycle} \\
\forall \TV_{i}  \in \Know_{i} , & \MechEA_{i}(\widetilde{\mixedA}_{i}^{(k)} \sqcup \pureA_{-i}) \TV_{i}  - \MechEP_{i}(\widetilde{\mixedA}_{i}^{(k)} \sqcup \pureA_{-i}) &
   \leq & \MechEA_{i}(\mixedA_{i}^{(k+1)} \sqcup \pureA_{-i}) \TV_{i}  - \MechEP_{i}(\mixedA_{i}^{(k+1)} \sqcup \pureA_{-i})
\end{array}
\end{equation*}
Thus, for any $z_{i} \in \Know_{i} \cap \widetilde{\Know}_{i}$, setting $\TV_{i}=\widetilde{\TV}_{i}=z_{i}$
we see that, for all $\pureA_{-i} \in \Strats_{-i}$ and for $k=1,2,\dots$
\begin{equation*}
\hfill
\begin{array}{rcll}
&\MechEA_{i}(\mixedA_{i}^{(k)} \sqcup \pureA_{-i}) z_{i} & - &\MechEP_{i}(\mixedA_{i}^{(k)} \sqcup \pureA_{-i}) \\
\leq\,&\MechEA_{i}(\widetilde{\mixedA}_{i}^{(k)} \sqcup \pureA_{-i}) z_{i} & - & \MechEP_{i}(\widetilde{\mixedA}_{i}^{(k)} \sqcup \pureA_{-i}) \\
\leq\,&\MechEA_{i}(\mixedA_{i}^{(k+1)} \sqcup \pureA_{-i}) z_{i} & - &\MechEP_{i}(\mixedA_{i}^{(k+1)} \sqcup \pureA_{-i})\enspace.
\end{array}
\hfill
\end{equation*}
That is, we have an infinite and non-decreasing sequence that is bounded from above. (Indeed, $z_{i}\leq B$, $\MechEA_{i}$ ranges between 0 and 1, and each price is non-negative.)
Thus, since $x_{i},y_{i} \in \Know_{i} \cap \widetilde{\Know}_{i}$, for any $\veps' > 0$ there exists a (sufficiently large) $k$ such that
\begin{equation*}
\hfill
\left|F_i^A(\sigma_i^{(k)}, s_{-i})x_i - F_i^P(\sigma_i^{(k)}, s_{-i}) - \left(F_i^A(\widetilde{\sigma}_i^{(k)},s_{-i})x_i - F_i^P(\widetilde{\sigma}_i^{(k)}, s_{-i})\right)\right| < \veps'
\hfill
\end{equation*}
and
\begin{equation*}
\hfill
\left|F_i^A(\sigma_i^{(k)}, s_{-i})y_i - F_i^P(\sigma_i^{(k)}, s_{-i}) - \left(F_i^A(\widetilde{\sigma}_i^{(k)},s_{-i})y_i - F_i^P(\widetilde{\sigma}_i^{(k)}, s_{-i})\right)\right| < \veps'
\hfill
\end{equation*}
Now consider the following two linear functions:
\begin{equation*}
g(z) \defeq F_i^A(\sigma_i^{(k)}, s_{-i})z - F_i^P(\sigma_i^{(k)}, s_{-i})
\quad\text{ and }\quad
h(z) \defeq F_i^A(\widetilde{\sigma}_i^{(k)},s_{-i})z - F_i^P(\widetilde{\sigma}_i^{(k)}, s_{-i})
\enspace.
\end{equation*}
We have showed  that $|g(x_{i})-h(x_{i})| < \veps'$ and $|g(y_{i}) - h(y_{i})| < \veps'$. We now use the fact that if two linear functions are close at two points, they must have similar slopes. In particular,
\begin{align*}
|\textrm{slope}(g) - \textrm{slope}(h)| &= \left| \frac{g(x_{i})-g(y_{i})}{x_{i}-y_{i}} - \frac{h(x_{i})-h(y_{i})}{x_{i}-y_{i}} \right| \\
&\leq \frac{|g(x_{i}) - h(x_{i})| + |g(y_{i}) - h(y_{i})|}{|x_{i}-y_{i}|} < \frac{2\veps'}{|x_{i}-y_{i}|}
\enspace.
\end{align*}
The proof is complete by taking $\veps' = \veps |x_{i} - y_{i}|/2$. \bqed

\section{Formal Statement and Proof of Theorem 3}
\label{sec:undominated-negative-det}


\textbf{Theorem 3.}\hspace{2mm}
\emph{
For all $n\geq 1$, $\Err\in(0,1)$,  $\Bound \geq \frac{5}{\Err}$, and all  \emph{deterministic} mechanisms $\Mech = (\Strats,\OutFunc)$ for $\Class_{n,\Bound,\Err}$, there exist a context $(n,\Bound,\Err, \Know,\theta)\in \Class_{n,\Bound,\Err}$  and a strategy profile $\pureA \in \UWD(\Know)$ such that:}

\begin{equation*}
\hfill
\label{equation-target}
\SW\big(\TV, \OutFunc(\pureA)\big)
\leq \left(\left(\frac{1-\Err}{1+\Err}\right)^{2}  + \frac{4}{\Bound}\right) \MSW(\TV) \enspace.
\hfill
\end{equation*}

\vspace{2mm}
\noindent
\emph{Proof.}
Choose $x$ and $y$ such that $\lfloor (1+\Err)x \rfloor = \Bound$ and $y \defeq \lfloor \frac{(1-\Err)x+2}{1+\Err} \rfloor$. Since $\Bound> \frac{5}{\Err}$, $\lceil (1-\Err) y \rceil$ belongs to $ \{0,1,\dots,\Bound\}$.
Furthermore, recalling that $\Int{x} \defeq [(1-\Err)x ,(1+\Err)x] \cap \{0,1,\dots,\Bound\}$, one can verify that  $\Int{x}$ and $\Int{y}$
both  contain the  two integers   $ \lceil x(1-\Err) \rceil$ and $ \lceil x(1-\Err) \rceil + 1$.

Choose $\veps$ such that $\frac{1}{n}+\veps<1$. Then (the Undominated Intersection)  \lemmaref{lemma:uwd_single_lemma_cycle} guarantees that
\begin{equation}\label{equationLemma1}
\forall i \in \Players \; \exists \mixedA_{i} \in \Dist(\UWD_{i}(\Int{x})) \text{ and } \mixedA_{i}' \in \Dist(\UWD_{i}(\Int{y})) \text{ such that } \forall s_{-i} \in \Strats_{-i}:
\end{equation}
\begin{equation*}
\hfill
\left\vert \MechEA_{i}(\mixedA_{i} \sqcup \pureA_{-i}) - \MechEA_{i}({\mixedA}_{i}' \sqcup \pureA_{-i}) \right\vert < \veps \enspace.
\hfill
\end{equation*}

Now consider the allocation distribution $\MechEA(\mixedA_{1}',\dots,\mixedA_{n}')$, where the randomness comes from the mixed strategy profile since $\Mech$ is a deterministic mechanism. Since the good will be assigned with a total probability mass of $1$, by averaging, there exists a player $j$ such that $\MechEA_{j}(\mixedA_{1}',\dots,\mixedA_{n}')\leq \frac{1}{n}$: that is, player $j$ wins the good with  probability at most $\frac{1}{n}$. Without loss of generality, let $j=1$. In particular, there exist
$\pureA_{-1}' \in \UWD_{2}(\Int{y})\times \cdots \times \UWD_{n}(\Int{y})$, such that  $\MechEA_{1}(\mixedA_{1}',\pureA_{-1}')\leq \frac{1}{n}$. This together with  \equationref{equationLemma1} implies  that $\MechEA_1(\mixedA_{1} \sqcup \pureA_{-1}') \leq \frac{1}{n}+\veps<1$. In turn, this implies that there exists a \emph{pure} strategy
$s_1 \in \UWD_{1}(\Int{x})$ such that, setting  $\pureA\defeq (\pureA_{1} \sqcup \pureA_{-1}')$,   $\MechEA_1(\pureA) = 0$.

Now we construct the desired $\Err$-approximate candidate-valuation profile $\Know$ and the true-valuation profile $\TV$ as follows:
\begin{equation*}
\hfill
\Know \defeq \big(\Int{x},\Int{y},\dots,\Int{y}\big) \quad \text{ and } \quad
\TV \defeq \big(\lfloor (1+\Err)x \rfloor, \lceil (1-\Err)y \rceil, \dots, \lceil (1-\Err)y \rceil\big)
\enspace.
\hfill
\end{equation*}

\noindent
Note that $\pureA \in \UWD(\Know)$, $\TV\in \Know$, and $\MSW(\TV)= \lfloor (1+\Err)x \rfloor $.
Since $\MechEA_1(\pureA) = 0$,
\begin{align*}
     \SW\big(\TV, \OutFunc(\pureA)\big)
&=    \lceil (1-\Err)y \rceil
\leq (1-\Err)y+1 \\
&\leq \frac{(1-\Err)^2x}{1+\Err} + 3 = \frac{(1-\Err)^2}{(1+\Err)^2}(1+\Err)x + 3 \\
&\leq \frac{(1-\Err)^2}{(1+\Err)^2} \lfloor (1+\Err)x \rfloor + 4 \leq \frac{(1-\Err)^2}{(1+\Err)^2} \lfloor (1+\Err)x \rfloor + \frac{4}{\Bound}\lfloor (1+\Err)x \rfloor \\
&\leq \left(\frac{(1-\Err)^2}{(1+\Err)^2} + \frac{4}{\Bound}\right) \lfloor (1+\Err)x \rfloor
=    \left(\frac{(1-\Err)^2}{(1+\Err)^2} + \frac{4}{\Bound}\right) \MSW(\TV)
\enspace.
\end{align*}
Thus the theorem holds. \bqed

\section{Formal Statement and Proof of Theorem 5}
\label{sec:undominated-negative-ran}

\textbf{Theorem 5.}\hspace{2mm}
\emph{
For all $n\geq 1$, $\Err\in(0,1)$,  $\Bound \geq \frac{5}{\Err}$, and all (deterministic or probabilistic) mechanisms $\Mech = (\Strats,\OutFunc)$ for $\Class_{n,\Bound,\Err}$, there exist a context $(n,\Bound,\Err,\Know,\TV)\in \Class_{n,\Bound,\Err}$ and a strategy profile $\pureA \in \UWD(\Know)$ such that
\begin{equation}\label{G-target}
\EXP\Big[\SW\big(\TV, \OutFunc(\pureA)\big)\Big]
\leq \left(\frac{(1-\Err)^2 + \frac{4\Err}{n}}{(1+\Err)^2} + \frac{4}{\Bound}\right) \MSW(\TV) \enspace.
\end{equation}
}

\vspace{2mm}
\noindent
\emph{Proof.}
(The first part of the proof  closely tracks that of Theorem 3 in \appendixref{sec:undominated-negative-det}.%
\footnote{Very informally, the only differences are that the allocation distribution $\MechEA(\mixedA_{1}',\dots,\mixedA_{n}')$ now depends also on the ``coin tosses of the mechanism", and that one can no longer guarantee the existence of  a pure strategy $\pureA$ such that  $\MechEA_1(\pureA)=0$.})

Choose $x$ and $y$ such that $\lfloor (1+\Err)x \rfloor = \Bound$ and $y \defeq \lfloor \frac{(1-\Err)x+2}{1+\Err} \rfloor$. Then again $\lceil (1-\Err) y \rceil$ belongs to $ \{0,1,\dots,\Bound\}$ and  $\Int{x}$ and $\Int{y}$ both  contain the following two integer points: $ \lceil x(1-\Err) \rceil$ and $ \lceil x(1-\Err) \rceil + 1$.

Since we always have $\lceil (1-\Err)y \rceil < (1-\Err)y + 1$, we can choose $\veps \in (0,1-\frac{1}{n})$ such that
\begin{equation*}
\frac{n-1}{n}\lceil (1-\Err)y \rceil + \veps \lceil (1-\Err)x \rceil - \veps \lceil  (1-\Err)y \rceil < \frac{n-1}{n}(1-\Err)y + 1 \enspace.
\end{equation*}
Then (the Undominated Intersection)  \lemmaref{lemma:uwd_single_lemma_cycle} guarantees that
\begin{equation}\label{equationLemma2}
\forall i \in \Players \text{ there exist } \mixedA_{i} \in \Dist(\UWD_{i}(\Int{x})) \text{ and } \mixedA_{i}' \in \Dist(\UWD_{i}(\Int{y})) \text{ such that } \forall s_{-i} \in \Strats_{-i}:
\end{equation}
\begin{equation*}
\left\vert \MechEA_{i}(\mixedA_{i} \sqcup \pureA_{-i}) - \MechEA_{i}({\mixedA}_{i}' \sqcup \pureA_{-i}) \right\vert < \veps \enspace.
\end{equation*}

Again consider the allocation distribution $\MechEA(\mixedA_{1}',\dots,\mixedA_{n}')$. By averaging, there exists some player $j$ such that $\MechEA_{j}(\mixedA_{1}',\dots,\mixedA_{n}')\leq \frac{1}{n}$.
Thus, by our choice of $\veps$ and \equationref{equationLemma2}, we have that $\MechEA_1(\mixedA_{1} \sqcup \mixedA_{-1}') \leq \frac{1}{n}+\veps$. This implies that there exists a \emph{pure} strategy profile $\pureA=(\pureA_{1} \sqcup \pureA_{-1}')$ that is in the support of $(\mixedA_{1} \sqcup \mixedA_{-1}')$
---and thus in $\UWD_{1}(\Int{x})\times \UWD_{2}(\Int{y})\times \cdots \times \UWD_{2}(\Int{y})$---
 such that $\MechEA_1(\pureA_{1} \sqcup \pureA_{-1}') \leq \frac{1}{n}+\veps$.
Now define 
\begin{align*}
\Know& \defeq \big(\Int{x},\Int{y},\dots,\Int{y}\big) \quad \text{ and } \quad
\TV \defeq \big(\lfloor (1+\Err)x \rfloor, \lceil (1-\Err)y \rceil, \dots, \lceil (1-\Err)y \rceil\big)
\enspace.
\end{align*}

Notice that $\pureA \in \UWD(\Know)$, $\TV\in \Know$, and $\MSW(\TV)= \lfloor (1+\Err)x \rfloor $.
We now show that $\pureA$, $\Know$, and $\TV$ satisfy the desired \equationref{G-target}:
\begin{align*}
     \EXP\Big[\SW\big(\TV,\OutFunc(\pureA_{1} \sqcup \pureA_{-1}')\big)\Big]
&\leq  \left(\frac{1}{n}+\veps\right) \cdot \lfloor (1+\Err)x \rfloor +\left(\frac{n-1}{n}-\veps\right) \cdot \lceil (1-\Err)y \rceil  \\
&= \frac{1}{n} \cdot \lfloor (1+\Err)x \rfloor + \frac{n-1}{n} \cdot \lceil (1-\Err)y \rceil + \veps \lceil (1-\Err)x \rceil - \veps \lceil  (1-\Err)y \rceil \\
&<  \frac{1}{n} \cdot \lfloor (1+\Err)x \rfloor + \frac{n-1}{n} \cdot (1-\Err)y + 1 \\
&\leq   \frac{1}{n} \cdot \lfloor (1+\Err)x \rfloor + \frac{n-1}{n} \cdot \frac{(1-\Err)^2x}{1+\Err} + 3 \\
&<    \frac{1}{n} \cdot \lfloor (1+\Err)x \rfloor +  \frac{n-1}{n} \cdot \frac{(1-\Err)^2}{(1+\Err)^2} \lfloor (1+\Err)x \rfloor +4 \\
&<    \left( \frac{1}{n} + \frac{n-1}{n} \cdot \frac{(1-\Err)^2}{(1+\Err)^2} + \frac{4}{\Bound} \right) \lfloor (1+\Err)x \rfloor \\
&=    \left( \frac{1}{n} + \frac{n-1}{n} \cdot \frac{(1-\Err)^2}{(1+\Err)^2} + \frac{4}{\Bound} \right) \MSW(\TV) \\
&=    \left(\frac{(1-\Err)^2 + \frac{4\Err}{n}}{(1+\Err)^2} + \frac{4}{\Bound} \right) \MSW(\TV)
\enspace.
\end{align*}
\bqed

\section{The Distinguishable Monotonicity Lemma}
\label{sec:proof-distinguishable-monotonicity-lemma}

%
%

Let us recall a traditional way to define auction mechanisms from suitable
allocation functions.

\begin{definition}
\label{def:f-allocation-mech}
If $\AlloFunc: [0,\Bound]^n \rightarrow [0,1]^n$ is an integrable%
\footnote{Specifically, we require that, for each $\BID_{-i}$, the function $f_{i}(z \sqcup \BID_{-i}$) is integrable with respect to $z$ on $[0,\Bound]$.}
 \eWPF, then we denote by $\Mech_{f}$ the mechanism $(\Strats,\OutFunc)$ where
 $\Strats = \{0,1,\dots,\Bound\}^n$ and $\OutFunc$ is so defined:  on input bid profile $\BID\in\Strats$,
\begin{itemize}[nolistsep]
\item
with probability $f_{i}(\BID)$ the good is assigned to player $i$, and

\item
if player $i$ wins, he pays $P_{i}= \BID_{i}  - \frac{\int_{0}^{\BID_{i}} \AlloFunc_{i}(\zz \sqcup \BID_{-i}) \,d\zz}{\AlloFunc_{i}(\BID_{i} \sqcup \BID_{-i})}$ (and all other players pay $P_j=0$ for $ j\neq i$.)
\end{itemize}

\end{definition}

\begin{remark}
$ $
\begin{itemize}[nolistsep]
\item $\Mech_{\AlloFunc}$ is deterministic if and only if $\AlloFunc(\{0,1,\dots,\Bound\}^n ) \subseteq \{0,1\}^n$.
\item For all player $i$ and bid profile $\BID$, the expected price $\MechEP_{i}(\BID)$ is equal to $\BID_{i}\cdot \AlloFunc_{i}(\BID_{i} \sqcup \BID_{-i})  - \int_{0}^{\BID_{i}} \AlloFunc_{i}(\zz \sqcup \BID_{-i}) \,d\zz$.

\item We stress that $\Mech_f$ continues to have the discrete strategy space $S=\{0,1,\dots,\Bound\}^n$. The analysis over a continuous domain for $f$ is only a tool for proving the lemma.

\item In the exact-valuation world, it is well known that a single-good auction mechanism $M$ is very-weakly dominant-strategy-truthful if and only if $M=M_f$ for some function $f$ that is (integrable and) {\em monotonic}, that is, such that each $f_{i}$ is non-decreasing in the bid of player $i$ for any fixed choice of bids of all other players.

\end{itemize}
\end{remark}

We now slightly strengthen the notion of monotonicity.

\begin{definition}
\label{def:allo_monotonicity}
Let $\AlloFunc: [0,\Bound]^{n} \rightarrow [0,1]^{n}$ be a \eWPF.
For $d \in \{1,2\}$, we say that $\AlloFunc$ is \defem{$d$-distinguishably monotonic} ($d$-DM, for short) if $\AlloFunc$ is integrable,  monotonic, and satisfying the following ``distinguishability'' condition:
\begin{equation*}
\forall\, i \in \Players\,,\, \forall \BID_{i},\BID_{i}' \in \Strats_{i} \textrm{ s.t. } \BID_{i}\leq \BID_{i}'-d,\, \exists\, \BID_{-i} \in \Strats_{-i}
\quad \int_{\BID_{i}}^{\BID_{i}'} \big(\AlloFunc_{i}(\zz \sqcup \BID_{-i}) - \AlloFunc_{i}(\BID_{i} \sqcup \BID_{-i})\big) \, d\zz > 0 \enspace.
\end{equation*}
If $\AlloFunc$ is $d$-DM, we say that $\Mech_{\AlloFunc}$ is $d$-DM.
\end{definition}

Distinguishability is certainly an additional requirement to monotonicity, but actually is a mild one. (Indeed, the second-price mechanism is $2$-DM and, if ties are broken at random, even $1$-DM.%
\footnote{For example, the second-price mechanism with lexicographic tie-breaking
is the mechanism $M_f$ where $f$ is defined as follows: $\forall i \in \Players$ and $\forall v \in \{0,\ldots,\Bound\}^n$,
\begin{equation*}
f_{i}(v) \defeq
             \left\{
              \begin{array}{ll}
                1, & \text{ if (a) } v_{i} > \max v_{-i} \text { or (b) } v_{i} = \max v_{-i} \text{ and } i = \min\{ j : v_j=v_{i}\}; \\
                0, & \text{ otherwise.}
              \end{array}
            \right.
\end{equation*}
To see that this mechanism is $2$-DM, consider two bids $v_{i}$ and $v_{i}'$ of player $i$ that are at least a distance of two apart; by choosing a strategy sub-profile for the other players where the highest bid falls between $v_{i}$ and $v_{i}'$, we can ensure that the desired integral is positive. A slightly more refined argument shows that the second-price mechanism breaking ties at random is $1$-DM.}%
)
Yet, in our Knightian setting, this mild additional requirement is quite useful for ``controlling'' the \emph{undominated strategies} of a mechanism, and thus for engineering implementations of desirable social choice functions in undominated strategies.

\begin{lemma}[\textbf{Distinguishable Monotonicity Lemma}]
\label{lemma:UWD-char}
If $f$ is a $d$-DM allocation function, then $\Mech_{\AlloFunc}$ is such that, for any player $i$ and $\Err$-approximate candidate-valuation profile $\Know$,
\begin{align*}
\UWD_{i}(\Know_{i}) &\subseteq \{\min \Know_{i}, \dots, \max \Know_{i} \} \enspace \quad \quad \quad \: \: \: \text{if $d=1$, and}\\
\UWD_{i}(\Know_{i}) &\subseteq \{\min \Know_{i}-1, \dots, \max \Know_{i}+1 \} \enspace\text{ if $d=2$.}
\end{align*}

\end{lemma}
(Above, $\min \Know_{i}$ and $\max \Know_{i}$ respectively denote the minimum and maximum integers in $\Know_{i}$.)

\vspace{2mm}
\noindent
\emph{Proof.}
For every $i \in [n]$, let $\BID_{i}^{\sbot} \defeq \min \Know_{i}$ and $\BID_{i}^{\stop} \defeq \max \Know_{i}$. Then, to establish our lemma it suffices to prove that, $\forall i \in [n]$ and $\forall d \in \{1,2\}$,  the following four properties hold:

\begin{itemize}\itemsep -1pt
\item[1.]
\emph{$\BID_{i}^{\sbot}$ very-weakly dominates every $\BID_{i} \leq \BID_{i}^{\sbot} - d$}.

\item[2.]
\emph{$\BID_{i}^{\stop}$ very-weakly dominates every $\BID_{i} \geq \BID_{i}^{\stop} + d$}.

\item[3.]
\emph{There is a strategy sub-profile $\BID_{-i}$ for which $\BID_{i}^{\sbot}$ is strictly better than every $\BID_{i} \leq \BID_{i}^{\sbot} - d$}.

\item[4.]
\emph{There is a strategy sub-profile $\BID_{-i}$ for which $\BID_{i}^{\stop}$ is strictly better than every $\BID_{i} \geq \BID_{i}^{\stop} + d$}.
\end{itemize}

\noindent
\emph{Proof of Property 1.}
Fix any (pure) strategy sub-profile $\BID_{-i} \in \Strats_{-i}$ for the other players and any possible true valuation $\TV_{i} \in \Know_{i}$. Letting $\BID^{\sbot} = ( \BID_{i}^{\sbot} \sqcup \BID_{-i})$ and $\BID = (\BID_{i} \sqcup \BID_{-i})$, we prove that
\begin{align*}
   & \EXP\Big[\Utility_{i}\big(\TV_{i},\OutFunc(\BID^{\sbot})\big)\Big] - \EXP\Big[\Utility_{i}\big(\TV_{i},\OutFunc(\BID)\big)\Big] \\
=\;& \big(\AlloFunc_{i}(\BID^{\sbot}) - \AlloFunc_{i}(\BID)\big) \cdot \TV_{i} - \big(\MechEP_{i}(\BID^{\sbot}) - \MechEP_{i}(\BID)\big) \\
=\;& \big(\AlloFunc_{i}(\BID^{\sbot}) - \AlloFunc_{i}(\BID)\big) \cdot \TV_{i} - \left(\BID_{i}^{\sbot}\cdot \AlloFunc_{i}(\BID^{\sbot}) - \int_{0}^{\BID_{i}^{\sbot}} \AlloFunc_{i}(\zz \sqcup \BID_{-i}) \,d\zz - \BID_{i}\cdot \AlloFunc_{i}(\BID) + \int_{0}^{\BID_{i}} \AlloFunc_{i}(\zz \sqcup \BID_{-i}) \,d\zz \right) \\
\DEBUG{
=\;&{\color{gray} \big(\AlloFunc_{i}(\BID^{\sbot}) - \AlloFunc_{i}(\BID)\big) \cdot \TV_{i} - \BID_{i}^{\sbot}\cdot \AlloFunc_{i}(\BID^{\sbot}) + \BID_{i}\cdot \AlloFunc_{i}(\BID) + \int_{\BID_{i}}^{\BID_{i}^{\sbot}} \AlloFunc_{i}(\zz \sqcup \BID_{-i}) \,d\zz} \\
=\;&{\color{gray} \big(\AlloFunc_{i}(\BID^{\sbot}) - \AlloFunc_{i}(\BID)\big) \cdot (\TV_{i}-\BID_{i}^{\sbot}) - (\BID_{i}^{\sbot}-\BID_{i}) \cdot \AlloFunc_{i}(\BID) + \int_{\BID_{i}}^{\BID_{i}^{\sbot}} \AlloFunc_{i}(\zz \sqcup \BID_{-i}) \,d\zz} \\
}
=\;& \big(\AlloFunc_{i}(\BID^{\sbot}) - \AlloFunc_{i}(\BID)\big) \cdot (\TV_{i}-\BID_{i}^{\sbot}) + \int_{\BID_{i}}^{\BID_{i}^{\sbot}} \big(\AlloFunc_{i}(\zz \sqcup \BID_{-i}) - \AlloFunc_{i}(\BID)\big) \, d\zz
\enspace.
\end{align*}
Now note that, since $\TV_{i} \in \Know_{i}$, $\TV_{i} - \BID_{i}^{\sbot} = \TV_{i} - \min \Know_{i} \geq 0$; moreover, by the monotonicity of $\AlloFunc$, whenever $\zz \geq \BID_{i}$, it holds that $\AlloFunc_{i}(\zz \sqcup \BID_{-i}) \geq \AlloFunc_{i}(\BID)$. We deduce that $\EU_{i}\big(\TV_{i},\OutFunc(\BID^{\sbot})\big) \geq \EU_{i}\big(\TV_{i},\OutFunc(\BID)\big)$. We conclude that $\BID_{i}^{\sbot}$ very-weakly dominates $\BID_{i}$.

\noindent
\emph{Proof of Property 2.}
Analogous to that of Property 1 and omitted.

\smallskip

\noindent
\emph{Proof of Property 3.}
Due to the $d$-distinguishable monotonicity of $\Mech$, $\BID_{i} \leq \BID_{i}^{\sbot} - d$ implies the existence of a strategy sub-profile $\BID_{-i}$
 making $\int_{\BID_{i}}^{\BID_{i}^{\sbot}}  \big(\AlloFunc_{i}(\zz \sqcup \BID_{-i}) - \AlloFunc_{i}(\BID)\big) \, d\zz$ strictly positive. For such  $\BID_{-i}$, therefore, playing $\BID_{i}^{\sbot}$ is strictly better than $\BID_{i}$.

\noindent
\emph{Proof of Property 4.}
Analogous to that of Property 3 and omitted.

Thus the lemma holds. \bqed
\section{Formal Statement and Proof of Theorem 2}
\label{sec:undominated-strategies-positive-det}

\textbf{Theorem 2.}
\hspace{2mm}
\emph{Let $\SPrice=(\SPriceS,\SPriceF)$ be the second-price mechanism with any deterministic tie-breaking rule. Then, for all contexts $(n,\Bound,\Err,\Know,\TV)$ and all strategy profiles $\BID \in \UWD(\Know)$:}
\begin{equation*}
\hfill
     \SW\big(\TV,\SPriceF(\BID)\big)
\geq \left(\frac{1-\Err}{1+\Err}\right)^{2} \MSW(\TV) - 2 \, \frac{1-\Err}{1+\Err}
\enspace.
\hfill
\end{equation*}

\vspace{2mm}
\noindent
\emph{Proof.}
Since $\Know$ is a $\Err$-approximate candidate-valuation  set,  for each player $i$ let $x_{i}$ be such that
$\Know_{i} \subseteq \Int{x_{i}}$. Then, in light of (the Distinguishable Monotonicity) \lemmaref{lemma:UWD-char}
and the previous observation that $\SPriceF^A$ is a 2-DM allocation function, we have that, for each player $i$:
\begin{equation}
\label{eq:thm2a-DML}
\UWD_{i}(x)  \subseteq \big\{ \lceil (1-\Err)x_{i} \rceil - 1 , \ldots, \lfloor (1+\Err)x_{i} \rfloor+1 \big\}  \enspace.
\end{equation}

Let $i^*$ be the player with the highest true valuation and $j^*$ the player winning the good under the bid profile $v$, that is, $\TV_{i^*}=\max_{i} \TV_{i}$  and  $\BID_{j^*}=\max_j \BID_j$.

If $i^* = j^*$ then we are done. If $i^*\neq j^*$, we need to show that $\TV_{j^*}$ is not much lower than $\TV_{i^*}$.

From \equationref{eq:thm2a-DML} we know that $ \lceil (1-\Err)x_{i^*} \rceil - 1 \leq \BID_{i^*} $ and $\BID_{j^*} \leq \lfloor (1+\Err)x_{j^*} \rfloor + 1$. Because $j^*$ is the winner, we also know that $\BID_{i^*}\leq \BID_{j^*}$. Combining these facts
and ``removing floors and ceilings" we have
$(1-\Err)x_{i^*} \leq(1+\Err)x_{j^*} +2$; equivalently,
$$
x_{j^*} \geq \frac{(1-\Err)}{(1+\Err)} x_{i^*} - 2 \frac{1}{(1+\Err)} \enspace .
$$
Since we also know that  $\TV_{j^*} \geq(1-\Err)x_{j^*}$ and $(1+\Err)x_{i^*}\geq \TV_{i^*}$, we obtain:
\begin{multline*}
\SW(\TV,\SPriceF(\BID)) = \TV_{j^*} \geq (1-\Err) x_{j^*}
\geq (1-\Err) \frac{1-\Err}{1+\Err} x_{i^*} - 2\frac{(1-\Err)}{(1+\Err)} \\
\geq (1-\Err) \frac{1-\Err}{1+\Err} \frac{1}{1+\Err} \TV_{i^*} -
2\frac{(1-\Err)}{(1+\Err)} = \frac{(1-\Err)^2}{(1+\Err)^2}\MSW(\TV) - 2\frac{(1-\Err)}{(1+\Err)} \enspace.
\end{multline*}
Thus, the theorem holds.\bqed

\begin{remark}
If $\SPrice=(\SPriceS,\SPriceF)$ were the second-price mechanism breaking ties at random (assigning a positive probability to each tie), then we can use a proof analogous to the one above, with the only difference being that $\SPriceF^A$ is $1$-DM (instead of only $2$-DM), and invoking the stronger inclusion of (the Distinguishable Monotonicity) \lemmaref{lemma:UWD-char}, to show the following, stronger lower bound:
\begin{equation*}
\hfill
     \EXP\Big[\SW\big(\TV,\SPriceF(\BID)\big)\Big]
\geq \frac{(1-\Err)^2}{(1+\Err)^2} \MSW(\TV)
\enspace.
\hfill
\end{equation*}
\end{remark}



\section{Formal Statement and Proof of Theorem 4}
\label{sec:undominated-strategies-positive-ran}

\textbf{Theorem 4.}\hspace{2mm}
\emph{$\forall n$, $\forall \Err \in (0,1)$, and $\forall \Bound$, there exists a mechanism $\MOPT=(S,F)$ such that for every $\Err$-approximate candidate-valuation profile $\Know$, every true-valuation profile $\TV\in\Know$, and every strategy profile $\BID \in \UWD(\Know)$:
\begin{equation*}
     \EXP\Big[\SW\big(\TV,F(\BID)\big)\Big]
\geq \left(\frac{(1-\Err)^2 + \frac{4\Err}{n}}{(1+\Err)^2}\right) \MSW(\TV)
\enspace.
\end{equation*}}

We break the construction of $\MOPT$ and its analysis into several steps. At the highest level, in order to leverage our Distinguishable Monotonicity Lemma, and thus choose $\MOPT= M_{\ffext}$ for a suitably chosen 1-DM allocation function $\ffext$.

\subsection{Our Allocation Function }

Given $\Err$, we find it natural to choose an allocation function $\ffext$ that is  {\em symmetric:} that is,  $\ffext (\BID')=\ffext (\BID)$ whenever the profile $\BID'$ consists of a permutation of the bids in $\BID$. In other words, ``renaming the players should not change the probability of allocating the good to a given player".

Also,  when some of the players' bids are much smaller than others, we find it intuitive to interpret the lower bids as being more likely to come from players with lower valuations. Accordingly,  our $\ffext$  gives positive probability only to the highest bids. However, when the highest bids are close to each other, we find it hard to  ``infer'' which one has been chosen by the player with the highest true valuation: after all, we are in a Knightian model. Therefore our $\ffext$ assigns the good to a randomly chosen high-bidding player. A bit more precisely, our $\ffext$ deterministically derives from the players' bids a {\em threshold}, and probabilistically chooses the winning player only among those bids lying above the threshold. To achieve optimality, however, one must be much more careful in allocating probability mass, and some complexities should be expected.
Let us now proceed more formally.

Let $D_\Err$ be the always positive quantity  defined as follows:  for all $\Err \in (0,1)$,
$$
\DD \defeq \left(\frac{1+\Err}{1-\Err}\right)^2-1 \enspace.
$$


\begin{definition}
\label{def:uwd_single_construct}
 For all $ \Err \in (0,1)$, define  $\ffext \colon [0,\Bound]^{n} \to [0,1]^{n}$ as follows: for every $i \in \Players$ and every $\zz = (\zz_{1},\dots,\zz_{n}) \in [0,\Bound]^{n}$
\begin{itemize}
  \item if $\zz_{1} \geq \zz_{2} \geq \cdots \geq \zz_{n}$, then
\begin{equation}\label{eqn:uwd_single_construct}
\ffext_{i}(\zz) \defeq
             \left\{
              \begin{array}{ll}
                \frac{1}{n} \cdot \frac{n+\DD}{n^{*}+\DD} \cdot \frac{\zz_{i}(n^{*}+\DD) - \sum_{j=1}^{n^{*}} \zz_j}{\zz_{i} \DD}, & \hbox{if $i \leq n^{*}$,} \\
                0, & \hbox{if $i > n^{*}$;}
              \end{array}
            \right.
            \enspace,
\end{equation}
where $n^{*}$ is the  index in $\{1,2,\dots,n\}$
such that
\begin{equation}\label{eqn:uwd_single_construct_pick12}
\zz_{1} \geq \cdots \geq \zz_{n^*} > \frac{\sum_{j=1}^{n^{*}} \zz_j}{n^{*} + \DD} \geq \zz_{n^*+1} \geq \cdots \geq \zz_n \enspace.
\end{equation}

  \item else, $\ffext_{i}(\zz) \defeq \ffext_{\pi(i)}(\zz_{\pi(1)},\dots,\zz_{\pi(n)})$ where $\pi$ is any permutation of the players such that $\zz_{\pi(1)} \geq \cdots \geq \zz_{\pi(n)}$.
\end{itemize}
We refer to  $\frac{\sum_{j=1}^{n^{*}} \zz_j}{n^{*} + \DD}$ as the \emph{\textbf{bid threshold}}, to players $1,\dots,n^{*}$ as the \emph{\textbf{candidate winners}}, and to the players $n^{*}+1,\dots,n$ as the \emph{\textbf{losers}}.
\end{definition}

\subsection{Our Allocation Function is Well Defined}

\begin{lemma}
\label{lemma:fd-is-well-defined}
$\ffext$ is an allocation function.
\end{lemma}

\vspace{2mm}
\noindent
\emph{Proof.}
Assume, without loss of generality, that $\zz_{1} \geq \zz_{2} \geq \cdots \geq \zz_{n}$.

We first prove that $n^{*}$ exists and is unique, and begin with its existence.

Note that there exists an integer $n'$ in $\Players$ such that
\begin{equation}\label{eqn:first}
\forall\, i > n', \quad \zz_{i} \leq \frac{\sum_{j=1}^{n'} \zz_j}{n' + \DD} \enspace.
\end{equation}
Indeed, \equationref{eqn:first} vacuously holds for $n'=n$. Now letting $n''$  be the least such integer,
the following two facts hold:
\begin{equation}\label{eqn:uwd_single_construct_pick1}
\forall\, i > n'', \quad \zz_{i} \leq \frac{\sum_{j=1}^{n''} \zz_j}{n'' + \DD}
\quad \text{ {\rm and }}
\end{equation}

\begin{equation}
 \exists k\geq n'' \text{ such that } \zz_{k} > \frac{\sum_{j=1}^{n''-1} \zz_j}{n''-1+\DD} \enspace.
\end{equation}

\noindent
Because $z$ is non-decreasing, the last inequality implies  $\zz_{n''} > \frac{\sum_{j=1}^{n''-1} \zz_j}{n''-1+\DD}$; equivalently, $\zz_{n''} > \frac{\sum_{j=1}^{n''} \zz_j}{n''+\DD}$.
Invoking again the monotonicity of $\zz$, we have
\begin{equation}\label{eqn:second}
\forall\, i \leq n'', \quad \zz_{i} > \frac{\sum_{j=1}^{n''} \zz_j}{n''+\DD} \enspace.
\end{equation}
Thus, \ref{eqn:first} and \ref{eqn:second}
imply that choosing $n^*=n''$ satisfies \equationref{eqn:uwd_single_construct_pick12}.

Next, we prove that $n^{*}$ is unique. Suppose by way of contradiction that there exist two integers $n^{\sbot}$ and $n^{\stop}$,  $n^{\sbot}<n^{\stop}$, both satisfying \equationref{eqn:uwd_single_construct_pick12}. Now define
$$S^{\sbot} \defeq \sum_{j=1}^{n^{\sbot}} \zz_j, \quad S^{\stop} \defeq \sum_{j=1}^{n^{\stop}} \zz_j, \quad S^{\mid} \defeq S^{\stop} - S^{\sbot},\text{\quad and }n^{\mid} \defeq n^{\stop}-n^{\sbot} \enspace.$$

\noindent
 Invoking \equationref{eqn:uwd_single_construct_pick12} with $n^{\sbot}$ and $n^{\stop}$, we deduce that for $i \in \{n^{\sbot}+1,\dots,n^{\stop}\}$,
\begin{equation*}
\frac{S^{\sbot}}{n^{\sbot}+\DD} \geq \zz_{i} > \frac{S^{\stop}}{n^{\stop}+\DD} = \frac{S^{\sbot}+S^{\mid}}{n^{\sbot}+n^{\mid}+\DD} \enspace.
\end{equation*}
Averaging over all $\zz_{i}$ such that $i \in \{n^{\sbot}+1,\dots,n^{\stop}\}$, we get
\begin{equation}
\label{eq:asdna}
\frac{S^{\sbot}}{n^{\sbot}+\DD} \geq \frac{S^{\mid}}{n^{\mid}} > \frac{S^{\sbot}+S^{\mid}}{n^{\sbot}+n^{\mid}+\DD} \enspace.
\end{equation}
Let us now show that the second inequality of \equationref{eq:asdna} contradicts the first:
\begin{align*}
\frac{S^{\mid}}{n^{\mid}} > \frac{S^{\sbot}+S^{\mid}}{n^{\sbot}+n^{\mid}+\DD}
&\Leftrightarrow (n^{\sbot}+n^{\mid}+\DD)S^{\mid} > n^{\mid}(S^{\sbot}+S^{\mid}) \nonumber \\
&\Leftrightarrow (n^{\sbot}+\DD)S^{\mid} > n^{\mid}S^{\sbot}
\Leftrightarrow \frac{S^{\mid}}{n^{\mid}} > \frac{S^{\sbot}}{(n^{\sbot}+\DD)} \label{eqn:uwd_single_construct_trick}
\enspace.
\end{align*}
The contradiction establishes the uniqueness of $n^*$.

Finally, to prove that $\ffext$ is an allocation function we must argue that (a)  $\ffext_{i}(\zz)\geq 0$ for every $i$ and $\zz$, and  (b) $\sum_{i} \ffext_{i}(\zz) \leq 1$ for every $\zz$.

Since we are assuming $\zz_{1} \geq \zz_{2} \geq \cdots \geq \zz_{n}$, inequality (a) holds because  $\ffext_i(z)=0$ for $i>n^*$ by definition, and because  \equationref{eqn:uwd_single_construct_pick12} tells us that $\zz_{i}(n^{*}+\DD) - \sum_{j=1}^{n^{*}} \zz_j > 0$ for  $i\leq n^*$.

As for inequality (b), since it is easy to see that  $\sum_{j=1}^{n^{*}} \frac{\zz_j}{\zz_{i}}\geq n^*$, we have
\begin{align*}
\sum_{i=1}^n \ffext_{i}(\zz)
&= \frac{1}{n} \cdot \frac{n+\DD}{n^{*}+\DD} \cdot \sum_{i=1}^{n^{*}} \frac{\zz_{i}(n^{*}+\DD) - \sum_{j=1}^{n^{*}} \zz_j}{\zz_{i} \DD} \\
&= \frac{1}{n} \cdot \frac{n+\DD}{(n^{*}+\DD)\DD} \cdot \left( n^{*}(n^{*}+\DD) - \sum_{i=1}^{n^{*}} \sum_{j=1}^{n^{*}}
\frac{\zz_j}{\zz_{i}} \right) \\
& \leq \frac{1}{n} \cdot \frac{n+\DD}{(n^{*}+\DD)\DD} \cdot \left( n^{*}(n^{*}+\DD) - n^{*} n^{*} \right)\\
&
= \frac{nn^*+n^*D_\Err}{n n^* + nD_\Err} \leq 1
\enspace.
\end{align*}
\wqed

\subsection{Our Allocation Function is 1-Distinguishably Monotonic}

\begin{lemma}
\label{lemma:fd-monotone}
$\ffext$ is monotonic.
\end{lemma}

\noindent
\emph{Proof.}
Since $\ffext$ is symmetric,  it suffices to show its  monotonicity with respect to a single coordinate, and we choose the $n$-th one for notational convenience.

Let $z_{-n}=(z_1,\ldots,z_{n-1}) \in [0,B]^{n-1}$ and assume, with no  generality loss, that $\zz_{1}\geq \zz_{2}\geq \cdots \geq \zz_{n-1}$.

We must prove that if  $0\leq \zz_n^{\sbot}  < \zz_n^{\stop} \leq \Bound$, then
\begin{equation}\label{eqn:monotonicty}
\ffext_n( \zz_{-n} \sqcup \zz_n^{\sbot}) \leq \ffext_n( \zz_{-n} \sqcup \zz_n^{\stop}) \enspace.
\end{equation}
To prove \equationref{eqn:monotonicty} we establish  two convenient claims.
Before doing so, we wish to stress that, although $z_{-n}$ is assumed to be monotonically non increasing, when  $z_n$ is chosen arbitrarily the profile $(\zz_{-n}, \zz_n)$
may not be monotonic.

\medskip

\noindent{\em CLAIM 1.}
{\em
If $n'$ is the number of candidate winners when only the first $n-1$ players are bidding and
their bid profile is $\zz_{-n}$, then for all $z_n$
\begin{align}
\label{eq:always-loser} \zz_{n} \leq \frac{\sum_{j=1}^{n'} \zz_j}{n' + \DD} &\;\Rightarrow\; \ffext_{n}(\zz_{-n}\sqcup \zz_n)=0 \enspace\text{(i.e., $n$ is a loser)} \\
\label{eq:always-winner} \zz_{n} > \frac{\sum_{j=1}^{n'} \zz_j}{n' + \DD} &\;\Rightarrow\; \ffext_{n}(\zz_{-n}\sqcup \zz_n)>0 \enspace\text{(i.e., $n$ is a winner)}
\end{align}
}
\medskip

\noindent
{\em Proof of CLAIM 1.}
The hypothesis of CLAIM 1 can be re-written as follows:
\begin{equation*}
\forall i\in\{1,2,\dots,n'\},\, \zz_{i} > \frac{\sum_{j=1}^{n'} \zz_j}{n' + \DD}; \text{ and } \forall i\in \{n'+1,\dots,n-1\},\, \zz_{i} \leq \frac{\sum_{j=1}^{n'} \zz_j}{n' + \DD}\enspace.
\end{equation*}
Let (``player $n$ join the game  bidding") $\zz_n \leq \frac{\sum_{j=1}^{n'} \zz_j}{n' + \DD}$. Then
\begin{equation*}
\forall i\in\{1,2,\dots,n'\}, \zz_{i} > \frac{\sum_{j=1}^{n'} \zz_j}{n' + \DD}; \text{ and } \forall i\in \{n'+1,\dots,n\}, \zz_{i} \leq \frac{\sum_{j=1}^{n'} \zz_j}{n' + \DD}\enspace.
\end{equation*}
That is,  the bid threshold continues to be $\frac{\sum_{j=1}^{n'} \zz_j}{n' + \DD}$, and the set of winners continues to be $\{1,2,\dots,n'\}$. Thus $n$ is a loser and inequality \ref{eq:always-loser} holds.

Let now (player $n$ join the game bidding) $\zz_{n} > \frac{\sum_{j=1}^{n'} \zz_j}{n' + \DD}$ and assume, for the sake of contradiction,  that $\ffext_n(\zz_{-n}\sqcup \zz_n)=0$, that is, that player $n$ is a loser.
Then, letting $n^*$ be the new number of candidate winners, by definition:
\begin{equation*}
\forall i\in\{1,2,\dots,n^*\}, \zz_{i} > \frac{\sum_{j=1}^{n^*} \zz_j}{n^* + \DD};\quad \forall i\in \{n^*+1,\dots,n\}, \zz_{i} \leq \frac{\sum_{j=1}^{n^*} \zz_j}{n^* + \DD}\enspace.
\end{equation*}

Thus, ``ignoring  $n$" we get
\begin{equation*}
\forall i\in\{1,2,\dots,n^*\}, \zz_{i} > \frac{\sum_{j=1}^{n^*} \zz_j}{n^* + \DD};\quad \forall i\in \{n^*+1,\dots,n-1\}, \zz_{i} \leq \frac{\sum_{j=1}^{n^*} \zz_j}{n^* + \DD}\enspace.
\end{equation*}

That is, $n^*$ is also the number of candidate winners under the hypothesis of CLAIM 1. Thus, the uniqueness of $n^*$ implies $n^*=n'$. In turn, this implies that
$\zz_n\leq \frac{\sum_{j=1}^{n'} \zz_j}{n' + \DD}$, a contradiction.
This contradiction proves that $n$ is a winner. Thus the claimed inequality \eqref{eq:always-winner} holds. \qed

\medskip

Thanks to CLAIM 1, to establish the monotonicity of $\ffext$ it suffices to prove that
\begin{equation}\label{eq:all-that-remains}
\text{{\bf if} } \frac{\sum_{j=1}^{n'} \zz_j}{n' + \DD} < \zz_{n}^{\sbot} < \zz_{n}^{\stop} , \text{ {\bf then} } \ffext_n( \zz_{-n} \sqcup \zz_n^{\sbot}) \leq \ffext_n( \zz_{-n} \sqcup \zz_n^{\stop}) \enspace.
\end{equation}
Notice that for such $\zz_{n}^{\sbot}$ and $\zz_{n}^{\stop}$, player $n$ is always a candidate winner.
Therefore, let $\{1,\dots,n^{\sbot},n\}$ and $\{1,\dots,n^{\stop},n\}$ be the winners when the bid profiles are $(\zz_{-n}\sqcup \zz_{n}^{\sbot})$ and $(\zz_{-n}\sqcup \zz_{n}^{\stop})$ respectively.
We now relate ${n}^{\sbot}$ and ${n}^{\stop}$.

\medskip

\noindent{\em CLAIM 2.} $\quad n^{\sbot}\geq n^{\stop}$.

\medskip

\noindent
{\em Proof of CLAIM 2.}
Assume by way of contradiction that $n^{\sbot} < n^{\stop}$.

\noindent
We proceed in a way similar to the proof of \lemmaref{lemma:fd-is-well-defined}. Set
\begin{center}
$n^{\mid} \defeq n^{\stop} - n^{\sbot}$, $\quad S^{\sbot} \defeq \sum_{j=1}^{n^{\sbot}} \zz_j, \quad$ and $\quad S^{\stop} \defeq \sum_{j=1}^{n^{\stop}} \zz_{j}= S^{\sbot}+S^{\mid}.$
\end{center}
Since $n^{\sbot}\leq i < n^{\stop}$ implies that player $i$ is a loser when the bid profile is $(\zz_{-n}\sqcup \zz_n^{\sbot})$ and a winner when the bid profile is $(\zz_{-n}\sqcup \zz_n^{\stop})$, we have
\begin{equation*}
\frac{S^{\sbot}+\zz_{n}^{\sbot}}{n^{\sbot}+1+\DD} \geq \zz_{i} > \frac{S^{\stop}+\zz_{n}^{\stop}}{n^{\stop}+1+\DD} = \frac{S^{\sbot}+S^{\mid}+\zz_{n}^{\stop}}{n^{\sbot}+n^{\mid}+1+\DD} \enspace.
\end{equation*}
Averaging over all $i$ such that $n^{\sbot}\leq i < n^{\stop}$ we get:
\begin{equation}\label{eqn:claim2}
\frac{S^{\sbot}+\zz_{n}^{\sbot}}{n^{\sbot}+1+\DD} \geq \frac{S^{\mid}}{n^{\mid}} >   \frac{S^{\sbot}+S^{\mid}+\zz_{n}^{\stop}}{n^{\sbot}+n^{\mid}+1+\DD} \enspace .
\end{equation}
Focusing on the second inequality of \eqref{eqn:claim2}, 
we have

\begin{align*}
\frac{S^{\mid}}{n^{\mid}} > \frac{S^{\sbot}+S^{\mid}+\zz_{n}^{\stop}}{n^{\sbot}+n^{\mid}+1+\DD}
&\Leftrightarrow (n^{\sbot}+n^{\mid}+1+\DD)S^{\mid} > n^{\mid}(S^{\sbot}+S^{\mid}+\zz_{n}^{\stop}) \nonumber \\
&\Leftrightarrow (n^{\sbot}+1+\DD)S^{\mid} > n^{\mid}(S^{\sbot}+\zz_{n}^{\stop}) \\
&\Leftrightarrow \frac{S^{\mid}}{n^{\mid}} > \frac{S^{\sbot}+\zz_{n}^{\stop}}{n^{\sbot}+1+\DD} \label{eqn:uwd_single_construct_trick}
\enspace.
\end{align*}
Thus, since $\zz_n^{\sbot} <\zz_n^{\stop} $, the second inequality \eqref{eqn:claim2} contradicts the first.

The contradiction establishes that $n^{\sbot}\geq n^{\stop}$ as claimed. \qed

\medskip

We now use the fact that $n^{\sbot}\geq n^{\stop}$ to prove   \equationref{eq:all-that-remains}, as desired.

If $n^{\sbot}=n^{\stop}$, then, for both $(\zz_{-n}\sqcup \zz_{n}^{\stop})$ and $(\zz_{-n}\sqcup \zz_{n}^{\sbot})$, the set of candidate winners is $\{1,2,\dots,n^{\sbot},n\}$. Thus, letting $n^*=n^{\sbot}+1=n^{\stop}+1$ be the number of candidate winners, we get
\begin{align*}
\ffext_{n}(\zz_{-n}\sqcup \zz_{n}^{\sbot}) &= \frac{1}{n} \cdot \frac{n+\DD}{n^{*}+\DD} \cdot \frac{\zz_{n}^{\sbot}(n^{*}+\DD) - \sum_{j=1}^{n^{*}-1} \zz_j - \zz_{n}^{\sbot}}{\zz_{n}^{\sbot} \DD} \\
&\leq \frac{1}{n} \cdot \frac{n+\DD}{n^{*}+\DD} \cdot \frac{\zz_{n}^{\stop}(n^{*}+\DD) - \sum_{j=1}^{n^{*}-1} \zz_j - \zz_{n}^{\stop}}{\zz_{n}^{\stop} \DD} = \ffext_{n}(\zz_{-n}\sqcup \zz_{n}^{\stop})
 \enspace.
\end{align*}
If $n^{\sbot} > n^{\stop}$, then let $n^{\sbot}=n^{\stop}+n^{\mid}$, $S^{\stop}=\sum_{j=1}^{n^{\stop}} \zz_j$ and $S^{\sbot}=\sum_{j=1}^{n^{\sbot}} \zz_j = S^{\stop}+S^{\mid}$ as before. Averaging over all $\zz_{i}$ such that $n^{\stop} < i \leq n^{\sbot}$ we get:
\begin{equation}
\frac{S^{\mid}}{n^{\mid}} > \frac{S^{\sbot}+\zz_{n}^{\sbot}}{n^{\sbot}+1+\DD}
= \frac{S^{\stop}+S^{\mid}+\zz_{n}^{\sbot}}{n^{\stop}+n^{\mid}+1+\DD} \enspace.
\end{equation}
But this is equivalent to 
\begin{equation}\label{eqn:uwd_single_construct_tmp}
\frac{S^{\mid}}{n^{\mid}} > \frac{S^{\stop}+\zz_{n}^{\sbot}}{n^{\stop}+1+\DD} \enspace.
\end{equation}

Letting $C_1 = \frac{n+\DD}{n}$ and $C_2 = C_1 \frac{1}{(n^{\sbot}+1+\DD)\zz_{n}^{\sbot}} \frac{1}{(n^{\stop}+1+\DD)\zz_{n}^{\stop}}$, we now do the final calculation:
\begin{align*}
&\ffext_{n}(\zz_{-n}\sqcup \zz_{n}^{\stop}) - \ffext_{n}(\zz_{-n}\sqcup \zz_{n}^{\sbot}) \\
&= C_{1} \cdot \Big( \frac{\zz_{n}^{\stop}(n^{\stop}+1+\DD) - S^{\stop}-\zz_{n}^{\stop}}{(n^{\stop}+1+\DD)\zz_{n}^{\stop} } - \frac{\zz_{n}^{\sbot}(n^{\sbot}+1+\DD) - S^{\sbot}-\zz_{n}^{\sbot}}{(n^{\sbot}+1+\DD)\zz_{n}^{\sbot} }\Big) \\
&= C_{1} \cdot \Big( \frac{S^{\sbot}+\zz_{n}^{\sbot}}{(n^{\sbot}+1+\DD)\zz_{n}^{\sbot} } - \frac{ S^{\stop}+\zz_{n}^{\stop}}{(n^{\stop}+1+\DD)\zz_{n}^{\stop} }\Big) \\
&= C_{2} \cdot \Big( (S^{\sbot}+\zz_{n}^{\sbot})(n^{\stop}+1+\DD)\zz_{n}^{\stop} - (S^{\stop}+\zz_{n}^{\stop})(n^{\sbot}+1+\DD)\zz_{n}^{\sbot} \Big) \\
&= C_{2} \cdot \Big( (S^{\stop}+S^{\mid}+\zz_{n}^{\sbot})(n^{\stop}+1+\DD)\zz_{n}^{\stop} - (S^{\stop}+\zz_{n}^{\stop})(n^{\stop}+n^{\mid}+1+\DD)\zz_{n}^{\sbot}\Big) \\
&= C_{2} \cdot \Big( S^{\stop}(n^{\stop}+1+\DD)(\zz_{n}^{\stop}-\zz_{n}^{\sbot}) + S^{\mid}(n^{\stop}+1+\DD)\zz_{n}^{\stop} - n^{\mid}(S^{\stop}+\zz_{n}^{\stop}) \zz_{n}^{\sbot} \Big) \\
&\geq C_{2} \cdot \Big( S^{\stop}(n^{\stop}+1+\DD)(\zz_{n}^{\stop}-\zz_{n}^{\sbot}) + S^{\mid}(n^{\stop}+1+\DD)\zz_{n}^{\stop} - n^{\mid}(S^{\stop}+\zz_{n}^{\sbot}) \zz_{n}^{\stop} \Big) \geq 0
\end{align*}
The last inequality has been derived using the fact that $\zz_{n}^{\stop}-\zz_{n}^{\sbot} \geq 0$ and (by \equationref{eqn:uwd_single_construct_tmp}) the fact that $S^{\mid}(n^{\stop}+1+\DD) - n^{\mid}(S^{\stop}+\zz_{n}^{\sbot}) > 0$.

This finishes the proof that $\ffext$ is monotonic.
\wqed

\begin{lemma}
\label{lemma:df-1dm}
$\ffext$ is $1$-distinguishably monotonic.
\end{lemma}

\noindent
\emph{Proof.}
We already know from \lemmaref{lemma:fd-monotone} that $\ffext$ is monotonic.

We now need to argue the integrability of $\ffext$, that is, that $\ffext_{i}(\zz_{i},\zz_{-i})$ is integrable in $\zz_{i}$ for any choice of $\zz_{-i}$. Again, because $\ffext$ is symmetric, it suffices to argue this for $i=n$. To do so, let us first point out that the number of winners does not increase when $\zz_n$ increases, for any $\zz_{-n}$. Indeed, letting $n'$ be the number of candidate winners when only the first $n-1$ players are bidding and
their bid profile is $\zz_{-n}$, CLAIM 1 in the proof \lemmaref{lemma:fd-monotone} tells us that when $\zz_{n} \leq \frac{\sum_{j=1}^{n'} \zz_j}{n' + \DD}$ player $n$ is always a loser and thus the number of winners remains constant; moreover, CLAIM 2 in the proof of \lemmaref{lemma:fd-monotone} tells us that when $\zz_{n}$ increases past $\frac{\sum_{j=1}^{n'} \zz_j}{n' + \DD}$ the number of winners does not decrease. Thus, $\ffext$ is piecewise continuous (as we have established that the number of winners does not increase when $\zz_n$ increases, the finite number of continuous pieces is at most $n$), and therefore $\ffext$ is integrable.

We are therefore left to prove the ``distinguishability condition''.

Fix a player $i \in \Players$ and two distinct valuations $\BID_{i},\BID_{i}' \in \{0,1,\dots,\Bound\}$, and assume that $\BID_{i}<\BID_{i}'$.
If we choose $\BID_{-i} = (\BID_{i},\BID_{i},\dots,\BID_{i})$, then:
\begin{itemize}
\item $\ffext_{i}(\BID_{i} \sqcup \BID_{-i})=\frac{1}{n}$ since there are $n$ winners, all bidding the same valuation, and
\item $\ffext_{i}(z \sqcup\BID_{-i})=\frac{1}{n\DD} (\DD+n-1-\frac{\BID_{i}}{z} (n-1)) > \frac{1}{n} $,
 when $\BID_{i} < z \leq (1+\DD)\BID_{i}$.
\end{itemize}
The upper bound $(1+\DD)\BID_{i}$ is to make sure that the number of winners is still $n$ on input $(\zz \sqcup \BID_{-i})$.
Thus, by definition, $\ffext_{i}(z \sqcup\BID_{-i})$ is a function that is \emph{strictly} increasing when $z$ increases in $(\BID_{i} , (1+\DD)\BID_{i}]$,
 and therefore
\begin{equation*}
\int_{\BID_{i}}^{\BID_{i}'} \big(\ffext_{i}(\zz \sqcup \BID_{-i}) - \ffext_{i}(\BID_{i} \sqcup \BID_{-i})\big) \, d\zz
\geq
\int_{\BID_{i}}^{\min\{\BID_{i}',(1+\DD)\BID_{i}\}} \big(\ffext_{i}(\zz \sqcup \BID_{-i}) - \ffext_{i}(\BID_{i} \sqcup \BID_{-i})\big) \, d\zz
> 0
 \enspace,
\end{equation*}
as desired.
\wqed

\subsection{Our Allocation Function is $\Err$-Good}

Our last lemma tells us that $\MOPT\defeq M_{\ffext}$ is a 1-distinguishably monotone mechanism. This property  (via our Distinguishable Monotonicity Lemma) simplifies the analysis of the undominated strategies for $\MOPT$, but otherwise has no bearing on proving  the social-welfare performance claimed for $\MOPT$ in Theorem 4. (As already remarked,  the probabilistic second-price mechanism is 1-DM, but only guarantees a fraction $\big(\frac{1-\Err}{1+\Err}\big)^2$ of the maximum social welfare.) Accordingly, our allocation function
$\ffext$ must and indeed does satisfy an additional property, {\em $\Err$-goodness}. We state this property below and  prove
that $\ffext$ indeed satisfies it. Only in the next section we shall prove the relevance of $\Err$-goodness for Theorem 4.

Recall that $\DD \defeq \left(\frac{1+\Err}{1-\Err}\right)^2-1$.

\begin{definition}\label{def:delta-good}
We say that a $1$-DM  \eWPF{} $\AlloFunc$ is \emph{\textbf{$\Err$-good}} if
$$
\label{eqn:uwd_single_sw_characterization}
\forall\, i\in \{1,\ldots,n\}, \; \forall\, \BID \in \{0,1,\dots,\Bound\}^n , \quad
\sum_{j=1}^{n} \AlloFunc_{j}(\BID)\BID_{j} + \DD \cdot \AlloFunc_{i}(\BID) \BID_{i} \geq \frac{1}{n} \cdot \BID_{i} (n+\DD) \enspace.
$$
\end{definition}

\begin{lemma}
\label{lemma:fd-is-good}
$\ffext$ is $\Err$-good.
\end{lemma}

\begin{proof}
As we already know that $\ffext$ is $1$-DM,  we establish the lemma by proving that the above inequality holds not only for the discrete cube $\{0,1,\dots,\Bound\}^n$, but actaually for the continuous cube $[0,\Bound]^n$.

Without loss of generality, assume $\zz_{1}\geq \zz_{2}\geq \cdots \geq \zz_{n}$.
Observe that:
\begin{align*}
  \sum_{i=1}^{n} \ffext_{i}(\zz) \zz_{i}
= \sum_{i=1}^{n^{*}} \ffext_{i}(\zz) \zz_{i}
&= \frac{1}{n} \cdot \frac{n+\DD}{n^{*}+\DD} \cdot \sum_{i=1}^{n^{*}} \frac{\zz_{i}(n^{*}+\DD) - \sum_{j=1}^{n^{*}} \zz_j}{\DD} \\
&= \frac{1}{n} \cdot \frac{n+\DD}{n^{*}+\DD} \cdot \left(\sum_{i=1}^{n^{*}} \zz_{i}\right)
\enspace.
\end{align*}
Then, for each looser $k$ (i.e., each player $k>n^{*}$,  we have
\begin{equation*}
     \sum_{j=1}^{n} \ffext_j(\zz)\zz_j + \DD \cdot \ffext_{k}(\zz) \zz_{k}
=    \sum_{i=1}^{n} \ffext_{i}(\zz) \zz_{i}
=    \frac{n+\DD}{n} \cdot \frac{\sum_{i=1}^{n^{*}} \zz_{i}}{n^{*}+\DD}
\geq \frac{n+\DD}{n} \cdot \zz_{k}
\end{equation*}
where the last inequality is due to  \equationref{eqn:uwd_single_construct_pick12}.

At the same time, for each winner $i$ (i.e., each player $i \leq n^{*}$), we have
\begin{align*}
  \sum_{j=1}^{n} \ffext_j(\zz)\zz_j + \DD \cdot \ffext_{i}(\zz) \zz_{i}
&= \frac{1}{n} \cdot \frac{n+\DD}{n^{*}+\DD} \cdot  \left(\sum_{i=1}^{n^{*}} \zz_{i}\right) + \DD \cdot \ffext_{i}(\zz) \zz_{i} \\
&= \frac{1}{n} \cdot \frac{n+\DD}{n^{*}+\DD} \zz_{i}(n^{*}+\DD)
= \frac{1}{n} \cdot \zz_{i} (n+\DD)
\end{align*}
again satisfying \equationref{eqn:uwd_single_sw_characterization}.
\end{proof}

\subsection{Proof of Theorem 4}

Since $\MOPT=\Mech_{\ffext}$, Theorem 4 immediately follows from  the fact that our allocation function $\ffext$ is $\Err$-good and the following

\begin{lemma}
If $\AlloFunc$ is $\Err$-good, then for the mechanism $\Mech_{\AlloFunc}=(S,F)$ we have: for all
contexts $(n,\Bound,\Err,\Know,\TV)$ and all strategy profile $\BID \in \UWD(\Know)$,
$$
     \EXP\Big[\SW\big(\TV,F(\BID)\big)\Big]
\geq \left(\frac{(1-\Err)^2 + \frac{4\Err}{n}}{(1+\Err)^2}\right) \MSW(\TV)
\enspace.
$$
\end{lemma}

\begin{proof}
Arbitrarily fix a context $(n,\Bound,\Err,\Know,\TV)$ and $\BID \in \UWD(\Know)$. Then, because
in any allocation the social welfare coincides with the true valuation of some player,
 to prove our lemma it  suffices to prove that
\begin{equation}
\forall i\in \Players, \quad    \sum_{j=1}^{n} \TV_{j} \AlloFunc_{j} (\BID)
\geq \left(\frac{(1-\Err)^2 + \frac{4\Err}{n}}{(1+\Err)^2} \right) \TV_{i} \enspace.
\label{eqn:uwd_single_sw_characterization_restate}
\end{equation}

For every  $i \in \Players$, let $x_{i} \in \real$ be such that $\Know_{i} \subseteq \Int{x_{i}}$, and let $\Int{x}=\Int{x_{1}}\times\cdots\times\Int{x_n}$.
Then,  the Distinguishable Monotonicity Lemma respectively implies $(1-\Err)x_{i} \leq \min \Know_{i} \leq \BID_{i} \leq \max \Know_{i} \leq (1+\Err)x_{i}$; equivalently,
$$
\frac{\BID_i}{1+\Err} \leq x_{i} \leq  \frac{\BID_i}{1-\Err} \enspace .
$$
Also, $\TV \in \Know$ implies $(1-\Err)x_{i} \leq \TV_{i} \leq (1+\Err)x_{i}$.
Combining the last two chains of inequalities yields
\begin{equation}
\label{eqn:uwd_single_sw_characterization_range}
\frac{1-\Err}{1+\Err}\BID_{i} \leq \TV_{i} \leq \frac{1+\Err}{1-\Err}\BID_{i} \enspace.
\end{equation}

Let us now argue that \equationref{eqn:uwd_single_sw_characterization_restate} holds by
arbitrarily fixing $\BID$ and $i$ and showing that it is impossible to construct a ``bad'' $\TV$  so as to violate \equationref{eqn:uwd_single_sw_characterization_restate}.

In trying to construct a ``bad'' $\TV$, it suffices to choose $\TV_{j}$ (for $j\neq i$) to be as small as possible, since $\TV_j$ only appears on the left-hand side with a positive coefficient.
For $\TV_{i}$, however, we may want to choose it as large as possible if $\AlloFunc_{i}(\BID) \geq \big(\frac{(1-\Err)^2 + \frac{4\Err}{n}}{(1+\Err)^2} \big)$, or as small as possible otherwise. So there are two extreme $\TV$'s.

Considering these extreme choices, we conclude that no $\TV$ contradicts \equationref{eqn:uwd_single_sw_characterization_restate} if:

\begin{equation*}
\sum_{j=1}^n \big(\frac{1-\Err}{1+\Err}\big)\BID_{j} \AlloFunc_{j} (\BID) \geq \left(\frac{(1-\Err)^2 + \frac{4\Err}{n}}{(1+\Err)^2} \right) \big(\frac{1-\Err}{1+\Err}\big) \BID_{i} \enspace,   \text{ and}
\end{equation*}

\begin{equation*}
\sum_{j=1}^n \big(\frac{1-\Err}{1+\Err}\big)\BID_{j} \AlloFunc_{j} (\BID) + \big(\frac{1+\Err}{1-\Err}-\frac{1-\Err}{1+\Err}\big) \BID_{i} \AlloFunc_{i}(\BID) \geq \left(\frac{(1-\Err)^2 + \frac{4\Err}{n}}{(1+\Err)^2} \right) \big(\frac{1+\Err}{1-\Err}\big) \BID_{i} \enspace.
\end{equation*}

Simplifying the above equations, \equationref{eqn:uwd_single_sw_characterization_restate} holds if both the following inequalities hold:
\begin{equation}
\sum_{j=1}^n \BID_{j} \AlloFunc_{j} (\BID) \geq
 \frac{n+\DD}{n} \cdot \frac{1}{\DD+1} \cdot \BID_{i}
 \enspace, \label{eqn:uwd_single_sw_characterization_extreme1}
\end{equation}

\begin{equation}
\sum_{j=1}^n \BID_{j} \AlloFunc_{j} (\BID) + \DD\hspace{-0.5mm} \cdot \BID_{i} \AlloFunc_{i}(\BID) \geq  \frac{n+\DD}{n} \BID_{i} \hspace{-0.5mm} \enspace.
\label{eqn:uwd_single_sw_characterization_extreme2}
\end{equation}
Note that \equationref{eqn:uwd_single_sw_characterization_extreme2} holds because it is implied by the hypothesis that $f$ is $\Err$-good; note also that \equationref{eqn:uwd_single_sw_characterization_extreme1} holds because it is implied by \equationref{eqn:uwd_single_sw_characterization_extreme2}. Indeed, since $ \frac{1}{\DD+1}=\big(\frac{1-\Err}{1+\Err}\big)^2<1$ for all $\Err \in (0,1)$,
\begin{equation*}
     \sum_{j=1}^n \BID_{j} \AlloFunc_{j} (\BID)
\geq \frac{1}{\DD+1} \left(\sum_{j=1}^n \BID_{j} \AlloFunc_{j} (\BID) + \DD \BID_{i} \AlloFunc_{i}(\BID) \right)
\geq \frac{1}{\DD+1}\frac{n+\DD}{n} \BID_{i}
\enspace.
\end{equation*}
Thus both \equationref{eqn:uwd_single_sw_characterization_restate} and both our lemma and Theorem 4 hold.
\end{proof}

\subsection{The Computational Efficiency of $\MOPT$}

Finally, we wish to clarify  that, although $\MOPT=M_{\ffext}=(S,F)$ is not as simple as the second-price mechanism, it can still be efficiently implemented.
%
%
That is, both the \eWPF{} $\MechEA = \ffext$
and the expected price function $\MechEP$ are efficiently computable over $\{0,1,\dots,\Bound\}^{n}$.

The computational efficiency of  $\MechEA$ is apparent once one realizes that the number of candidate winners, $n^*$, can be determined in linear time.

The computational efficiency of $\MechEP$ requires a bit of an argument. Without loss of generality, let us show how to compute the expected price for player $n$. Recall that, for a bid profile $\BID$,
\begin{equation*}
\MechEP_n(\BID_{-n}\sqcup \BID_{n}) = \ffext_{n}(\BID_{-n}\sqcup \BID_{n})\cdot \BID_{n} - \int_{0}^{\BID_{n}} \ffext_{n}(\BID_{-n} \sqcup \zz) \,d\zz \enspace.
\end{equation*}

When $\BID_{-n}$ is fixed, $\ffext_{n}$ is a function piece-wisely defined according to $\BID_{n}$, since different values of $\BID_{n}$ may result in different numbers of winners $n^*$. Assume without loss of generality that $\BID_{1}\geq \BID_{2}\geq \cdots \geq \BID_{n-1}$, and let $n'$ be the number of winners when player $n$ is absent.

When $\BID_{n} \leq \frac{\sum_{j=1}^{n'} \BID_j}{n'+\DD}$, the proof of the monotonicity of $\ffext$ implies that $\ffext_{n}=0$, so that integral below this point is zero.

When $\BID_{n} > \frac{\sum_{j=1}^{n'} \BID_j}{n'+\DD}$, one can again see from the proof of the monotonicity of $\ffext$ that $n^{*}$ is non-increasing as a function of $\BID_{n}$. Therefore, $\ffext_{n}$ contains at most $n$ different pieces and, for each piece with $n^{*}$ fixed, $\ffext_{n}(\BID_{-n}\sqcup \BID_{n}) = a + b/\BID_{n}$ is a function that is symbolically integrable. Therefore, the only question is how to calculate the pieces for $\ffext_{n}$.

Conceptually, one  starts from $\BID_{n} = \frac{\sum_{j=1}^{n'} \BID_j}{n'+\DD}$ and ``moves $\BID_{n}$ upwards", recording the points at which \equationref{eqn:uwd_single_construct_pick12} is violated, because these are the ``borders of the continuous pieces" of $\ffext_{n}$. Practically, this seemingly infinite procedure  may be efficiently carried out by a {\em line sweep method}.


\section{Conclusions}

Mechanism design is undoubtedly a fascinating field. One can only marvel at the possibility that an ignorant social planner can leverage the knowledge and the rationality of the players in order to obtain the outcomes he desires. But if we want to transform this beautiful theory into strong guarantees in the real world it is important to minimize its underlying assumptions.

To us, the assumption that each player knows ``on the nose''  his own valuation appears to be too idealized in many an environment. Even the assumption that each player knows a probability distribution from which his own true valuation has been drawn is  very strong.
To be safe, we should budget for the possibility of  Knightian players.

In any field, as we progress from idealized to more and more realistic models, we should expect to face additional complexities. Knightian mechanism design will be no exception. Nonetheless, we should remain optimistic. At least for single-good auctions, Knightian mechanism design is workable.

\section*{Acknowledgements}

We are grateful to Alan Deckelbaum for pointing out simplifications to our proofs.

\clearpage
\appendix

\section{Performance Diagrams}
\label{sec:diagrams}

\begin{figure}[h!]
\centering
\subfloat[][With $n=2$ players, the second-price mechanism performs \emph{worse} than randomly assigning the good for $\Err > 0.18$.]{\label{fig:singlegood1}\includegraphics[width=0.45\textwidth]{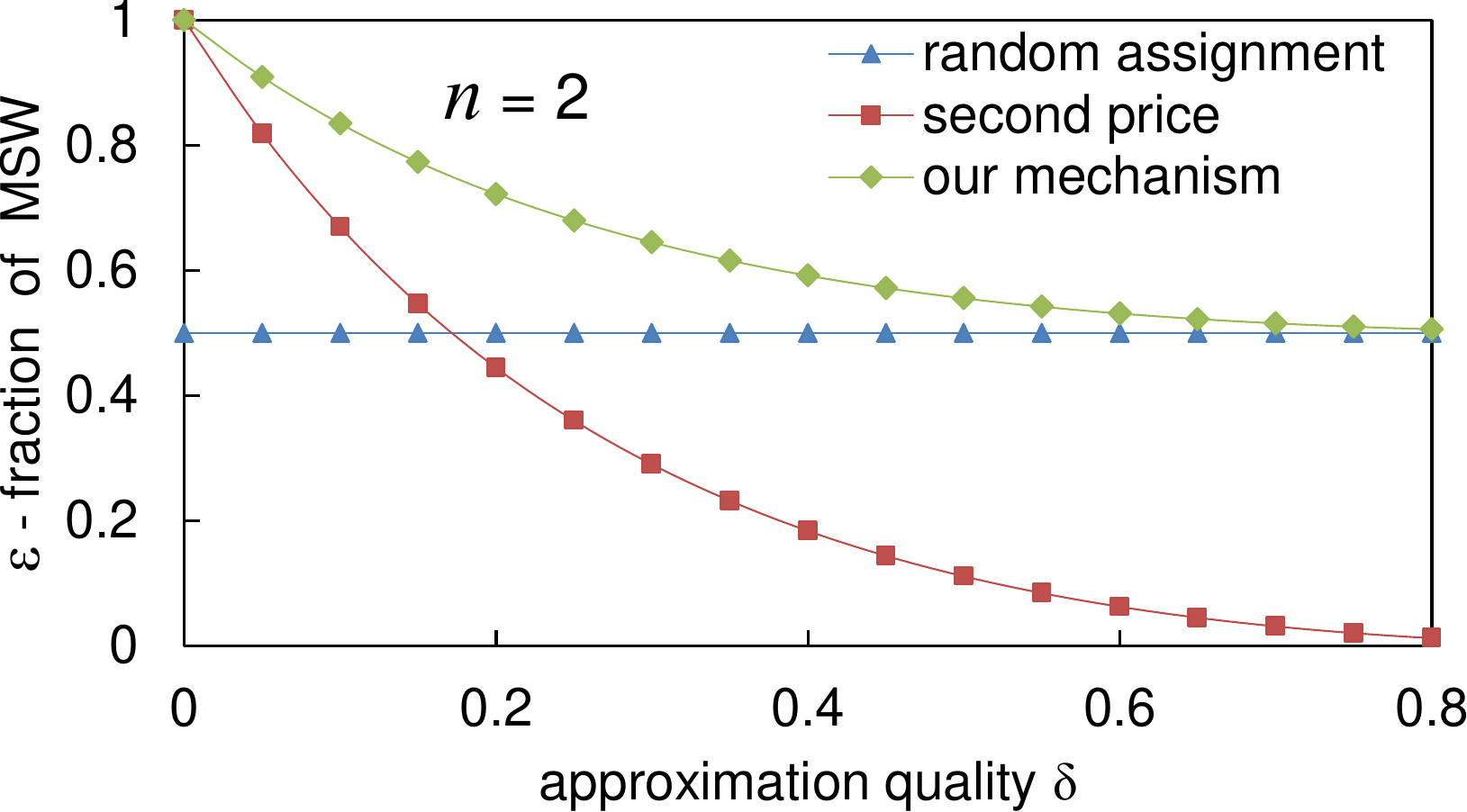}}
\qquad
\subfloat[][With $n=4$ players, the second-price mechanism performs \emph{worse} than randomly assigning the good for $\Err > 0.34$.]{\label{fig:singlegood2}\includegraphics[width=0.45\textwidth]{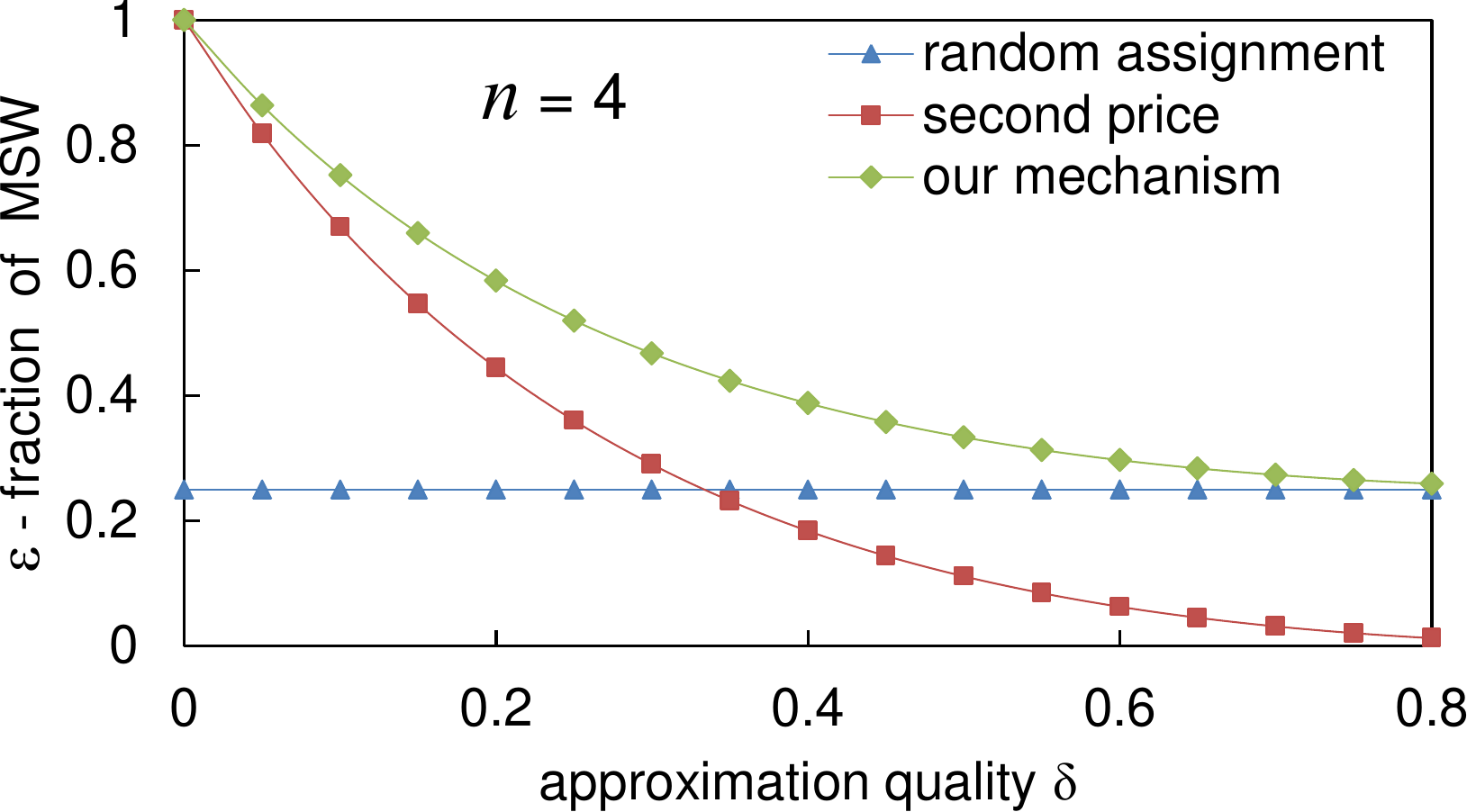}}
\qquad
\subfloat[][With $\Err=0.15$, the second-price mechanism always performs better than randomly assigning the good.]{\label{fig:singlegood3}\includegraphics[width=0.45\textwidth]{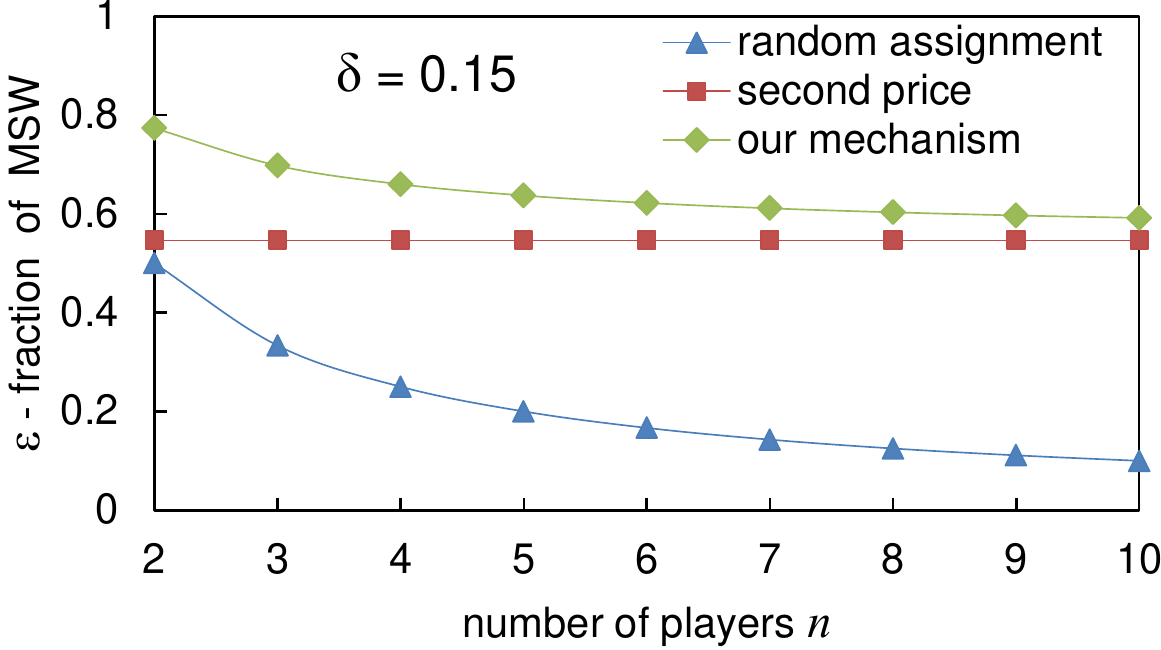}}
\qquad
\subfloat[][With $\Err=0.3$, the second-price mechanism performs \emph{worse} than randomly assigning the good for $n=2,3$.]{\label{fig:singlegood4}\includegraphics[width=0.45\textwidth]{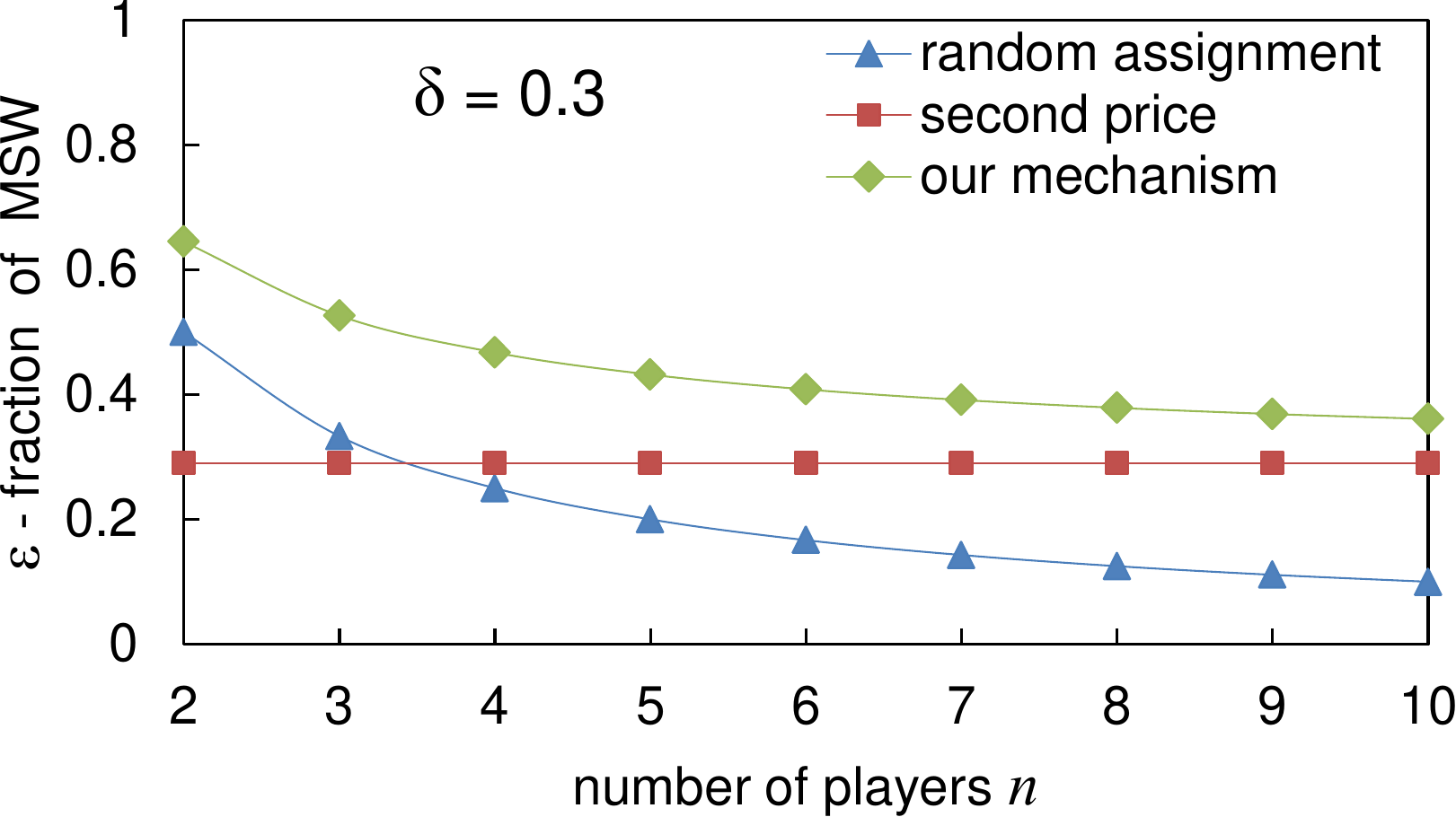}}
\caption{The social welfare guarantees of randomly assigning the good ($\veps = {\scriptscriptstyle\frac{1}{n}}$), the second-price mechanism ($\veps = {\scriptscriptstyle \frac{(1-\Err)^2}{(1+\Err)^2}}$),
and our optimal mechanism ($\veps ={\scriptscriptstyle \frac{(1-\Err)^2 + \frac{4\Err}{n}}{(1+\Err)^2}}$).
In (\ref{fig:singlegood1}) and (\ref{fig:singlegood2}) we compare $\veps$ versus $\Err$, and in (\ref{fig:singlegood3}) and (\ref{fig:singlegood4}) we compare $\veps$ versus $n$. The green data, our mechanism, is always better (at times significantly) than the other two mechanisms.}
\label{fig:diagrams}
\end{figure}

\newpage

\bibliographystyle{plainnat}
\bibliography{md-bib}

\end{document}